\documentclass[12pt,onecolumn,draftclsnofoot]{IEEEtran}

\usepackage{amsmath}
\usepackage{amsthm,amssymb,amsmath,bm}
\usepackage{subfigure}
\usepackage{amsfonts}
\usepackage{amsmath}

\usepackage{epsfig}
\usepackage{amssymb}
\usepackage{amsmath}
\usepackage{cite}
\usepackage[utf8x]{inputenc}
\graphicspath{{Figure/}}
\usepackage[nodisplayskipstretch]{setspace}   
\usepackage{multirow}
\usepackage{rotating}
\usepackage{graphicx}
\usepackage{tabularx}
\usepackage{array}
\usepackage{color,soul}
\usepackage{bm}
\usepackage{tikz}
\usepackage{blindtext}

\newtheorem{theorem}{Theorem}

\newtheorem{Defi}{Definition}
\newtheorem{remark}{Remark}
\newtheorem{Pro}{Proposition}

\DeclareMathOperator*{\argmax}{arg\,max} 
\usepackage{subfigure}
\usepackage{moreverb}
\usepackage{epsfig}
\usepackage{amsmath,amssymb,amsthm,mathrsfs,amsfonts,dsfont}
\usepackage{adjustbox,lipsum}
\usepackage{epsfig}
\usepackage{subfigure}
\usepackage{multirow}
\usepackage{rotating}
\usepackage{graphicx}
\usepackage{subfigure}
\usepackage{tabularx}
\usepackage{array}
\usepackage{anyfontsize}
\usepackage{color,soul}
\usepackage{graphicx,dblfloatfix}
\usepackage{epstopdf}
\usepackage{blindtext}
\usepackage{amsmath}
\usepackage{amsthm,amssymb,amsmath,bm}
\usepackage{subfigure}
\usepackage{amsfonts}
\usepackage{epsfig}
\usepackage{amssymb}
\usepackage{amsfonts}
\usepackage{amsmath}
\usepackage{cite}
\usepackage{graphicx}
\usepackage{fancyhdr}
\usepackage[subfigure]{tocloft}
\usepackage[font={small}]{caption}
\usepackage{tabularx}
\usepackage{tcolorbox}
\usepackage{cite}
\usepackage{algpseudocode}
\usepackage[linesnumbered,ruled,vlined]{algorithm2e}
\SetKwInput{KwInput}{Input}
\SetKwInput{KwOutput}{Output}
\usepackage{amsthm,amssymb,amsmath,bm}
\usepackage{graphicx}
\usepackage{fancyhdr}
\usepackage{subfigure}
\usepackage[subfigure]{tocloft}
\usepackage[font={small}]{caption}
\usepackage{subfigure}
\usepackage{tabularx}
\allowdisplaybreaks
\usepackage{cite}
\usepackage{url}
\usepackage{hyperref}

\newtheorem{Lem}{Lemma}
\def\blue{\textcolor{blue}}

\begin{document}
\title{
\huge
Minimizing the AoI
in Resource-Constrained Multi-Source Relaying Systems: Dynamic and Learning-based Scheduling
}
   \author{\IEEEauthorblockN{Abolfazl Zakeri, Mohammad Moltafet,  Markus Leinonen, and 
  \IEEEauthorblockN{Marian Codreanu}
  }
\thanks{
A. Zakeri, M. Moltafet, and M. Leinonen are with Centre for Wireless Communications--Radio Technologies,
  University of Oulu, Finland,
 e-mail: \{abolfazl.zakeri, mohammad.moltafet, markus.leinonen\}@oulu.fi. 
 M. Codreanu is with Department of Science and Technology,
  Link\"{o}ping University, Sweden, 
  e-mail: marian.codreanu@liu.se.
}
}	
\maketitle
\vspace{- 3.5 em}
\begin{abstract}
\footnote{Preliminary results of this paper were published in \cite{zakeri2021minimizing, Zakeri_Relay_Spw}.}
We consider a multi-source relaying system where independent sources randomly generate status update packets which are sent to the destination with the aid of a relay through unreliable links. 
We develop transmission scheduling policies to minimize the \blue{weighted} sum average age of information (AoI) subject to transmission capacity and long-run average resource constraints. 
We formulate a stochastic control optimization problem and solve it using a constrained Markov decision process (CMDP) approach and a drift-plus-penalty method.
The CMDP problem is solved by transforming it into an MDP problem using the Lagrangian relaxation method. 
We theoretically analyze the structure of optimal policies for the MDP problem 
and subsequently propose a structure-aware algorithm that returns a practical near-optimal policy.
Using the drift-plus-penalty method, we devise a near-optimal low-complexity policy that performs the scheduling decisions dynamically.
\blue{We also develop a model-free deep reinforcement learning policy for which 
the Lyapunov optimization theory and a dueling double deep Q-network are employed.
The complexities of the proposed policies are analyzed.}
Simulation results are provided to assess the performance of our policies and validate the theoretical results. The results show up to $91$\% performance improvement compared to a baseline policy. 
\end{abstract}


\vspace{- 1.5 em}
	\section{Introduction}
In many emerging applications  of wireless communications 
such as the Internet-of-Things (IoT), 
cyber-physical systems, and intelligent transportation systems, the freshness of  status information is crucial \cite{AoI_Mag, AoI_Monograph_Modiano}. The \textit{age of information} (AoI) has been proposed to characterize the information freshness in status update systems
\cite{Roy_2012}. The  AoI is defined as the time elapsed since the latest received status update packet was generated \cite{Roy_2012, AoI_Monograph_Modiano}. Recently, the AoI has attracted much interest 
in different areas, e.g., queuing systems \cite{MOhammad_1,Ahmad_Bd_Shroff,Roy_3,Marian_Information}, and scheduling and sampling problems \cite{Deniz_Relay,Brancu+2Hop,Eyton_Modiano, Eytan_Sch2, 2Hop_TWC, Elif_ARQ, Yin_Sun_2, Song_Time_varingCh, HetTraffic_JSAC, Age_Energy_Tradeoff_CMDP,Age_Learning_Autonomous_Scheduling, hatami, WS_Non_Uni_PS,  AoI_Throughput_Multi_Hop_ACM, RF_MultiHop_2022}. 
The reader can refer to  \cite{Roy_Survey} for a  survey on the AoI.
	  \\\indent
	  In some status update systems, there is no direct communication link between the source of information and the intended destination, or direct communication is costly.
In such systems, deploying an intermediate node, \blue{called \textit{relay}}\footnote{\textcolor{blue}{This relay could be a static node \cite{Deniz_Relay} or a mobile node, e.g., unmanned aerial vehicle (UAV) \cite{AoI_UAV_mag,AoI_UAV_V2V,  UAV_AoI_learnig, AoI_UAV_energy, UAV_AoI_nurul, Ferdowsi_relay} or a vehicle in the vehicular communications \cite{Vehicule_Book}. 
For instance, in \cite{UAV_AoI_nurul}, multiple UAVs serve as mobile relays between the sensors
and the base station, and the  goal is to optimize  the UAVs' trajectories  to minimize the average AoI and energy consumption. 
}},
is indispensable to  enable a \textit{long-distance communication}.  
 Deploying such a node has an array of benefits,
  e.g., saving on power usage of wireless sensors
and improving the transmission success probability.   
However, minimizing the AoI is particularly challenging in such \textit{relaying systems} due to need of \textit{jointly optimizing}  scheduling on both  source  and relay sides, especially in a \textit{multi-source} setup \cite{Deniz_Relay,Brancu+2Hop}. 
Moreover, minimizing the AoI becomes more challenging in the presence of unreliable wireless connectivity due to the possibility of losing some updates \cite{Eytan_Sch2}. 
At the same time, in practical status update systems, the number of transmissions is limited due to resource constraints (power, bandwidth, etc.), especially in power-limited sensor networks \cite{Brancu+2Hop,2Hop_TWC,Elif_ARQ, Deniz_Relay}.
\\\indent
In this paper, we consider a \textit{multi-source}  relaying status update system with \textit{stochastic arrivals}. The sources independently generate different types
 of status update packets which randomly arrive at a \textit{buffer-aided} transmitter. The transmitter sends the packets to a buffer-aided \textit{\blue{full-duplex}} relay  which further forwards the packets to the destination. 
 The buffers store the last arrived packet from each source. All  transmission links (channels), i.e., the transmitter-to-relay and relay-to-destination links, are 
 \textit{unreliable (error-prone)} and have a limited transmission capacity.
A practical application for  the considered system could be industrial monitoring, where status updates of various sensors in a given factory zone are first gathered by a low-power transmitter and then sent to a remote monitoring center with help of a relay. Another application could emerge in vehicular networks,  where status updates about various physical processes related to a vehicle are sent to a controller (e.g., a road side unit) for supporting vehicle safety applications \cite{Vehicule_Book}. However, the vehicle is far from the coverage of the controller, and thus, a relay (could be another vehicle \cite{Vehicule_Book}, or a UAV \cite{AoI_UAV_V2V}) is needed to establish the communication.
\\\indent
We formulate a \textit{stochastic control optimization problem} aiming  to minimize
the \blue{weighted} sum  average AoI (AAoI)  subject to transmission capacity constraints and 
a \textit{long-run average resource constraint}, which limits the average number of all transmissions in the system.
We develop three different  \textit{(transmission) scheduling polices} by solving the problem.
Namely, we provide: (1) a deterministic policy, (2) a drift-plus-penalty-based scheduling policy (DPP-SP), and (3) a deep reinforcement learning policy.  
\blue{A constrained Markov decision process (CMDP) approach and a drift-plus-penalty method are proposed.}
\blue{For the former, we first show that the unichain structure holds for the CMDP problem and then  apply the Lagrangian relaxation method to solve it.}  
\blue{We  \textit{theoretically} analyze the structure of an optimal policy for the resulting MDP problem and subsequently   propose a  \textit{structure-aware} algorithm
that provides a near-optimal deterministic policy
(which is an optimal policy for the MDP problem)
and
another deterministic policy that  gives a lower bound on the optimal value of the CMDP problem.}
We note  that  an optimal policy can be obtained by randomizing the proposed near-optimal deterministic policy and the lower-bound deterministic policy; however, obtaining such randomized policy might be computationally intractable. 
\blue{In  the drift-plus-penalty method, we transform
the main problem into a sequence of per-slot problems and 
then devise a near-optimal \textit{low-complexity}
DPP-SP,
which  performs the scheduling dynamically,   using a scheduling rule described  by a  \textit{closed-form} solution to the per-slot optimization problem. 
 Moreover, we  provide a model-free deep reinforcement learning algorithm for which 
 we first employ the Lyapunov optimization theory to transform the main problem into an MDP
problem and then adopt 
a dueling double deep Q-network (D3QN)\footnote{\blue{This method integrates double
deep Q-network (DQN) and dueling DQN 
to further alleviate the
overestimation problem of DQN and improve its  convergence rate 
\cite{D3QN_silver, D3DQN_First}.  Moreover, it {was} shown, e.g., in \cite{D3QN_II_TWC, D3QN_silver},
that D3QN generally gives better performance than the other two mentioned methods.}} to solve it.}
The proposed learning-based policy addresses the case in which  the packet arrival rates and the error probabilities of the wireless channels are not known a priori, i.e., so-called \textit{unknown environment}.  
It should be noted that 
the environment model may not be
(readily) available, or using a perfect model is not applicable in practice owing
to computational difficulties.
\blue{The computational complexity  of the proposed policies is analyzed.}
Finally,   extensive numerical analysis are provided to validate the theoretical results and show the effectiveness of the proposed scheduling policies.
\vspace{-1.5 em}
\subsection{Contributions}
The main contributions of this paper are summarized as follows:
\begin{itemize}
    \item We study the AoI in a multi-source buffer-aided \blue{full-duplex} relaying status update system
with stochastic arrivals and unreliable links. 
We formulate a stochastic optimization problem
that aims  to minimize
\blue{the weighted sum  AAoI}  subject to transmission capacity constraints and 
a long-run time average resource constraint.
\item We develop three different scheduling polices by solving the main optimization  problem. 
Particularly, we propose the CMDP approach and the drift-plus-penalty method.
Moreover, we develop a  deep reinforcement learning  algorithm by combining the Lyapunov optimization theory and D3QN. 
\item \blue{We theoretically analyze the structure of an optimal policy of the MDP problem (obtained via the Lagrangian relaxation) and develop a structure-aware \blue{iterative} algorithm for solving the CMDP problem.} \blue{The convergence of the algorithm is also proven.}
\blue{\item We devise a dynamic  near-optimal low-complexity scheduling policy, i.e., DPP-SP,  by providing 
a closed-form solution to the per-slot  problem obtained under the drift-plus-penalty method. 
Moreover, we prove that DPP-SP satisfies the average resource constraint.
\item We analyze the computational complexity of the proposed scheduling policies.} 
\item We  provide numerical analysis to verify the  theoretical results and assess the effectiveness of the devised  policies. The results show up to $91$\% performance improvement compared to a greedy-based baseline policy. 
\end{itemize}
\vspace{-1.5 em}
\subsection{Related Works}
	  	Recently, the  AoI  in   relaying systems has been studied in, e.g., \cite{MOradi_R,Deniz_Relay,Relay_Nikoss,Shroff,2Hop_ARQ, 2Hop_TWC,Relay_SA, T2R_Relay, Age_Relay_Stochastic_Any, Statis_Gua, AoI_Throughput_Trade, AoI_Energy_Relay, AoI_relay_Short_Packet}. 
	  	The authors of \cite{MOradi_R} analyzed the AoI in a discrete-time Markovian system for two different relay settings and analyzed the impact of relaying on the AoI. In \cite{Relay_Nikoss}, the authors analyzed   the AAoI   in a two-way relaying system under the \textit{generate-at-will} model (i.e., possibility of generating a new update at any time) model in which two sources exchange status data.
	  	The AoI performance under different policies (e.g., a last-generated-first-served policy) in general multi-hop single-source  networks  was studied in \cite{Shroff}.
In \cite{2Hop_TWC}, the authors studied the AoI in a single-source energy harvesting relaying system with error-free channels and designed offline and online age-optimal policies.
Reference \cite{2Hop_ARQ} analyzed the AAoI in a single-source relaying system with and without the automatic repeat-request technique, where results show that the automatic repeat-request technique can reduce the AAoI.
The age-energy tradeoffs in a relay-aided  status update system were studied in \cite{AoI_Energy_Relay}, where the expressions for the AAoI and average energy cost were derived.
In \cite{Age_Relay_Stochastic_Any}, the expression of the AoI distributions  in a single-source relaying system under different circumstances were derived.
Minimization of the AAoI through optimizing the blocklengths in short-packet communications in decode-and-forward relaying IoT networks was conducted in \cite{AoI_relay_Short_Packet}.
Authors of \cite{Statis_Gua} optimized the steady-state  AoI violation probability with respect to the sampling rate of monitoring a process in both single-hop and two-hop systems.
In \cite{Relay_SA}, the authors  considered a  single-source relaying system under  stochastic packet arrivals where 
	  	 the source communicates with the destination either through the direct link or via a relay.
	  	They proposed two different relaying protocols 
	  	and derived the respective  AAoI expressions. 
    
	  	In summary, only a few works, such as \cite{2Hop_TWC,Brancu+2Hop, Deniz_Relay}, have incorporated  a resource constraint  
	  	(as we do in this paper) when analyzing and/or optimizing  the AoI in a relaying system. 
	  	 Moreover, different from our multi-source system, most of the discussed works, e.g., \cite{2Hop_TWC,MOradi_R,Shroff,Relay_SA,2Hop_ARQ,Statis_Gua,Age_Relay_Stochastic_Any, AoI_relay_Short_Packet,AoI_Energy_Relay, T2R_Relay}, consider single-source relaying systems.
     \blue{Clearly, multi-source scheduling is generally substantially challenging, especially when there are also  resource constraints (as in this paper). Because one needs to properly allocate
     a limited amount of resources among multiple
sources, taking into account each source’s characteristics (e.g., the arrival rate of each source
and the source’s information importance). }
\\\indent Our relaying system, considered as a \textit{two-hop} network, is an extension of work
\cite{Eyton_Modiano}, where the authors provided scheduling policies for minimizing the AAoI in a \textit{one-hop buffer-free} network with stochastic arrivals and  an \textit{error-free} link  with no 
  average resource  constraint. 
\blue{In contrast,  our two-hop network is a buffer-aided network with error-prone links.}
The most-related works to our paper are  \cite{Brancu+2Hop, Deniz_Relay}.
	  	The  work \cite{Deniz_Relay} studied the AoI minimization in a multi-source relaying system with the generate-at-will model 
    and
	  	 unreliable channels.
	  	The authors proved that the greedy policy is an optimal scheduling policy for a setting called the error-prone symmetric IoT network whereas  for the general setting, they applied DQN. 
\blue{  In \cite{Brancu+2Hop},
	  	the authors studied  the AAoI minimization problem in a single-source \textit{half-duplex} relaying system with  the generate-at-will model under a constraint  on the average number of forwarding transmissions at the relay.
          In contrast  to \cite{Brancu+2Hop},  we consider a \textit{multi-source} setup;  because of the single-source setup in \cite{Brancu+2Hop}, the scheduling problem  of \cite{Brancu+2Hop} is essentially 
       the problem of optimizing whether the relay should receive or transmit at each slot,
       whereas our problem is multi-source scheduling.
Different to \cite{Eyton_Modiano}, we have two-dimensional decision variables in our problem which makes constructing optimal/good scheduling policies more  difficult.
We further consider an average resource constraint so that  our problem is a CMDP problem, whereas the problems of \cite{Deniz_Relay, Eyton_Modiano} are  MDP problems. Notably, not only solving a CMDP problem 
is substantially challenging but analyzing its optimal policy structure
is also challenging.
  Furthermore, the stochastic arrival model considered in our setup generalizes the generate-at-will model in \cite{Brancu+2Hop,Deniz_Relay} and brings additional challenges   in 
  the  design and analysis of scheduling policies since
  the statistics of the arrivals and
  the AoI at the
transmitter are also involved  in the
system dynamics.
       }
      \blue{Finally, besides the MDP/CMDP approach proposed in \cite{Eyton_Modiano, Brancu+2Hop,Deniz_Relay},
       we also propose the two different scheduling policies, i.e., DPP-SP, and the deep reinforcement learning policy that copes with unknown environments.}
        Even though \cite{Brancu+2Hop} also develops a low-complexity double threshold relaying policy, the thresholds need to be optimized numerically. In contrast, our low-complexity  scheduling policy requires to execute two simple operations.
        \vspace{-1.5 em}
\subsection{Organization} 
The rest of this paper is organized as follows. The system model and problem formulation are presented in Section \ref{SM_PF}.  The CMDP formulation and its solution  are presented in Section \ref{Sec_CMDP_Lag}. 
The 
DPP-SP is presented in Section \ref{Sec_LC_MW}.
The deep reinforcement learning algorithm  is provided in Section \ref{Sec_Learning}. 
\blue{The computational complexity of the proposed policies  is analyzed in Section \ref{Sec_ComAna}.}
The numerical analysis and conclusions are provided in Section \ref{Sec_Numerical_Res} and Section \ref{Sec_Conclusion}, respectively. 
	
\vspace{-1em}
\section{System Model and Problem Formulation}\label{SM_PF}
	\subsection{System Model}
We consider a status update system  consisting of a set $\mathcal{I}=\{1,\dots,I\}$ of $I$ independent sources,
		a buffer-aided transmitter\footnote{\textcolor{blue}{Even though there is one transmitter,  the  system is mathematically equivalent to one 
where each source 
directly sends its updates to the relay using a shared channel
and where  at most one
source is allowed to occupy the channel at each slot.} }, a \textcolor{blue}{ buffer-aided full-duplex}
		relay, and a destination, as depicted in   ${ \text{Fig.}~ \ref{SM}}$.
 \blue{The sources model physically separated fully autonomous 
sensors (i.e., they cannot be controlled and commanded) where their (status update)
packets are sent to the transmitter using a random access protocol.
Thus, the stochastic arrivals model is used to account for possible random packet losses on the links
between the sources and the transmitter due to, e.g., collisions, and/or for possible idle slots where the source sensors do not send updates.}
Additionally,  there is no direct communication link between the transmitter and the destination, and thus, the transmitter sends all status update packets to the destination via  the relay.
We assume that each status update is encapsulated in one packet\footnote{A status update packet
of each source contains a time stamp representing the time when the sample was generated and the
measured value of the monitored process.}.
The buffer size is one packet per source and
each buffer stores the most recently arrived packet of a source, as they contain the freshest information. More specifically, a packet of a source arriving at the transmitter replaces the packet of the same source in the transmitter's buffer; similarly, a packet of a source received by the relay replaces the packet of the same source in the relay's buffer. 
It is worth noting that considering one packet size buffer for each source is sufficient in our system, as storing and transmitting outdated packets does not improve the AoI. 
\blue{Moreover, the relay transmits the packet available in the buffer at the beginning of the slots, while the buffer is updated at the end of the slots (if a new packet is successfully received).}
	\\\indent
	We consider a discrete-time system with unit time slots $t\in\{0,1,2,\ldots\}$.
The sources, indexed
by $i\in\mathcal{I}$,  independently generate status update packets  according to the Bernoulli  distribution with parameter $\mu_i$. 
\textcolor{blue}{Note that $\mu_i=1$ gives the same performance  when considering  the  system with the generate-at-will model and no sampling cost.}
Let $u_i[t]$  be a binary indicator that shows whether a packet from  source  $i$  arrives at the transmitter  at the beginning of slot $t$, i.e., $u_i[t]=1$ indicates that a packet arrived; otherwise, $u_i[t]=0$. Accordingly, $\Pr\{u_i[t]=1\}=\mu_i$.
\textcolor{blue}{For clarity, the definitions of the  main symbols are collected in Table \ref{Table_Not}.}
\begin{figure}
    \centering
    \includegraphics[width=.6\textwidth]{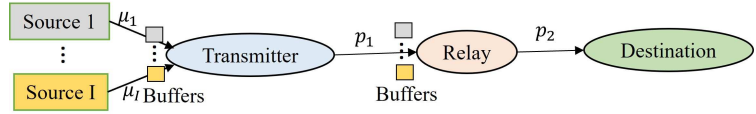}
    \vspace{-1 em}
    \caption{A multi-source relaying status update system  in which different status updates arrive at random time slots at the transmitter, which then sends the packets to the destination via a buffer-aided relay over unreliable links.
    }
    \label{SM}
    \vspace{- 1 em}
\end{figure}

\begin{table}
\small
\centering
\caption{ \textcolor{blue}{The key symbols with their definitions used in the paper }}
\vspace{-0.5 em}
	\begin{tabular}{ c |  c }
		\hline
		\textbf{Notation(s)}  & \textbf{Definition(s)}  \\
		\hline
		$ \mathcal{I}/I/i $ & \text{The set/number/index of sources }
		\\
		\hline
	$\theta_i[t]/\psi_i[t]/\delta_i[t]$  &  The AoI of source $i$ at the transmitter/relay/destination
 		\\
		\hline
	$\mu_i$  &  The arrival rate of source $i$
		\\
		\hline
	$w_i$  &  The weight of source $i$
 \\
		\hline
	$ p_1/p_2 $ 	&  \text{The reliability of the transmitter-relay/relay-destination link }
 		\\
		\hline
	$ \rho_1[t]/\rho_2[t] $ 	&  \text{The successful packet reception indicator of  the transmitter-relay/relay-destination link }
		\\
		\hline
  $ \alpha[t]/\beta[t]$  & The transmission decision at the transmitter/relay
  \\\hline
  $N$  &  The bound of the AoI  values
  \\\hline
  $\Gamma^{\mathrm{max} }$  & The transmission budget 
  \\\hline
  $\overline{K}$  & The average 
  number of total transmissions
  \\\hline
     $\overline{\delta}$  & The  weighted sum average AoI at the destination
     \\\hline
	\end{tabular}
 \label{Table_Not}
 \vspace{-2.5 em}
\end{table}
 
  \textit{Wireless Channels:} 
  As the wireless channels fluctuate over time, reception of updates (both by the relay and the destination) are subject to errors.  However, unsuccessfully received packets can be retransmitted; we assume that all retransmissions have the same reception  success probability.
  Let $p_1$ and $p_2$ be the successful transmission probabilities of  the transmitter-relay and relay-destination links, respectively. 
  Also, let $\rho_1[t]$ be a binary indicator of a successful packet reception by the relay
  in slot $t$, i.e.,  $\rho_1[t] = 1$ indicates that the transmitted packet is successfully received by the relay; otherwise,  $\rho_1[t] = 0$. 
  Similarly, let $\rho_2[t]$   be a binary indicator of a successful packet reception by the destination in slot $t$, i.e.,  $\rho_2[t] = 1$ indicates that the transmitted packet is successfully received by the destination; otherwise,  $\rho_2[t] = 0$.
We have $\Pr\{\rho_1[t]=1\}=p_1$ and  $\Pr\{\rho_2[t]=1\}=p_2$.
We assume that  perfect feedback (i.e., instantaneous  and error-free) is available for each link, and \blue{there is no interference between the links}\footnote{
\blue{Similarly to \cite{Deniz_Relay,Ferdowsi_relay}, we assume for simplicity that  the transmitter-relay and relay-destination links are using orthogonal frequency resource blocks.
Note that without this assumption, the performance of the full-duplex relay would highly depend on the internal self-interference cancellation mechanism \cite{Marian_relay} and this is beyond the scope of this paper.}
}.
 \\\indent \textit{Decision Variables:} 
 \textcolor{blue}{
 We assume that at most one packet transmission per slot is possible  over each link. 
 }
Let $\alpha[t]\in\{0,1,\dots,I\}$  denote the (transmission) decision of the transmitter in slot $t$, where $\alpha[t]=i$, $i\in\mathcal{I}$, 
 means that the transmitter sends the packet of source $i$ to the relay, and $\alpha[t]=0$ means that the transmitter stays idle.
 Similarly, let  $\beta[t]\in\{0,1,\dots,I\}$ denote the (transmission) decision of  the relay in slot $t$, where $\beta[t]=i$, $i\in\mathcal{I}$,
 means that the relay forwards the packet of source $i$ to the destination, and  $\beta[t]=0$ means that the relay stays idle. 
  We assume that there is a centralized controller  performing the scheduling.
\\\indent \textit{Age of Information:} Let $\theta_i[t]$ denote the AoI  of source $i$ at  the transmitter in slot $t$. Also, let
$\psi_i[t]$ denote the AoI  of source $i$ at the relay 
and $ \delta_i[t]$ denote the AoI  of source $i$ at the destination in slot $t$.
We make a common assumption (see e.g., \cite{Nikos_2020_RL,hatami,WS_Non_Uni_PS}) that AoI values are upper-bounded by a finite value $N$. Besides tractability, this accounts for the fact that once the available information about the process of interest becomes excessively stale, further counting would be irrelevant.
The evolution of the AoIs of each source ${ i\in\mathcal{I} }$ is given by
\begin{equation}
\begin{array}{ll}\label{Eq_AoIDy}
  &  \theta_i[t+1]=\left\{ 
  \begin{array}{ll}
     0,  & \text{if}~~ u_i [t+1]=1, \\
        \min\big(\theta_i[t]+1,N\big), & \text{otherwise},
    \end{array}
    \right.
\\&
    \psi_i[t+1]=\left\{ 
    \begin{array}{ll}
   \min\big(\theta_i[t]+1, N\big) &\text{if}~~ \alpha[t]=i,~\rho_1[t]=1,
    \\
    \min\big(\psi_i[t]+1, N \big), & \text{otherwise},
        \end{array}
    \right.
\\
&  \delta_i[t+1]=\left\{
\begin{array}{ll}
    \min\big(\psi_i[t]+1, N \big), &\text{if}~~
    \beta[t]=i,~\rho_2[t]=1,   
    \\
    \min\big(\delta_i[t]+1, N\big), &\text{otherwise}.
        \end{array}
    \right.
    \end{array}
\end{equation}
\vspace{-1.5 em}
\textcolor{blue}{
\begin{remark}
When $N$ is not
sufficiently large, the system performance
without bounding the AoI will be different from that with bounded AoI. 
The appropriate choice of $N$ depends  
on the system parameters
such as the number of sources and the links’ reliabilities.
\end{remark}
}
\vspace{-1.5 em}
\subsection{Problem Formulation}
\textcolor{blue}{
We denote  the \textit{weighted sum  average AoI at the destination} (WS-AAoI) by $\overline{\delta}$
and the \textit{average number of total transmissions per slot}  in the system by $\overline{K}$, which are  defined as 
  \begin{equation}
     \begin{array}{ll}
    & \displaystyle \overline{\delta} \triangleq\limsup_{T\rightarrow \infty} ~\frac{1}{T}
        \textstyle
        \sum_{t=1}^{T} 
         \mathbb{E}\left\{
          \textstyle \sum_i 
         \textcolor{blue}{ w_i }
          \delta_i[t]
        \right\},
        \\& \nonumber
        \displaystyle\overline{K}\triangleq\limsup_{T\rightarrow \infty} ~\frac{1}{T}  \textstyle\sum_{t=1}^{T} 
     \mathbb{E}\left\{ \mathds{1}_{\{\alpha[t]\neq0\}}+\mathds{1}_{\{\beta[t]\neq 0\}}\right\},
      \end{array}
 \end{equation}
    where 
    ${ w_i > 0,\,\forall i },$ denotes the weight of source $i$;
   $\mathds{1}_{\{\cdot\}}$ is an indicator function  which   equals to $1$ when the condition in $\{\cdot\}$ holds; 
   and
   $\mathbb{E}\{\cdot\}$ is the expectation with respect to the system randomness (i.e., random wireless channels and packet arrival processes) and the (possibly random) decision variables $\alpha[t]$ and $\beta[t]$\footnote{We assume that the decision variables $\alpha[t]$ and $\beta[t]$ are determined  based on the past and current AoI values; thus, we consider a set of control (scheduling) policies that contains  all causal policies \cite{Ahmad_Bd_Shroff}.}.
   }
   \textcolor{blue}{Moreover, $\overline{K} $ represents the system-wide power consumption.}
   By these definitions, we formulate 
the following stochastic optimization problem 
       \begin{subequations}
       \label{Org_P1}
       \begin{align}
        		\underset{\{\alpha(t),\beta(t)\}_{t=1,2,\ldots}}{\mbox{minimize}}~~~   &
        		\overline{\delta}
        		\vspace{-1em}
        		\\
        		\mbox{subject to}~~~ &
        	 \overline{K}
       \le \Gamma_{\max}, 
                   \label{Con_Sam1}
                   \end{align}
        		\end{subequations}  
 where the real value   $\Gamma_{\max}\in(0,2]$ is the maximum allowable average number of transmissions per slot in the system. 
 \textcolor{blue}{
 The time average constraint \eqref{Con_Sam1} represents  a system-wide power utilization budget.
 Thus, problem \eqref{Org_P1} provides a trade-off between the WS-AAoI and  the system-wide power consumption.
 Note that since the maximum number of per-slot transmissions in the system is $2$, we have $\Gamma_{\max}\in(0,2]$; the values  $\Gamma_{\max}\ge 2$ make the constraint inactive.
 }
Moreover, the \textit{Slater condition} \cite[Eq. 9.32]{Eitam_CMDP} clearly holds for  problem \eqref{Org_P1}, i.e., there  exists some 
set of decisions for which  $\bar{K}
       < \Gamma_{\max}$.
\\\indent
 In the next section, we will  present a CMDP
approach  to solve problem \eqref{Org_P1}. 

 \vspace{-1em}
	\section{ CMDP Approach to Solve Problem \eqref{Org_P1}}\label{Sec_CMDP_Lag}
In this section, we 
recast problem \eqref{Org_P1} into 
a CMDP problem 
which is then solved
by using the Lagrangian relaxation method. 
 \vspace{-1 em}
\subsection{CMDP Formulation}\label{Sec_CMDP_1}
We specify  the CMDP by the following elements:
\\$\bullet $
    \textit{State}:
      The  state of  the CMDP  incorporates the knowledge about all the AoI values in the system. 
We define the state in slot $t$ by
 $ \bold{s}[t]\triangleq (\theta_1[t],x_1[t],y_1[t],\dots,\theta_I[t],x_I[t],y_I[t])$,
 where $x_i[t]\triangleq \psi_i[t]-\theta_i[t],~\forall \,i\in\mathcal{I},$ and $y_i[t]\triangleq \delta_i[t]-\psi_i[t],~\forall \,i\in\mathcal{I},$ are the \textit{relative AoIs} at the relay and the destination in slot $t$, respectively. 
 Using the relative AoIs simplifies 
 the subsequent analysis and derivations. 
 The intuition is that 
 the evolution of the AoI of source $i$ at the destination  
 from slot $t$ to
  $t+1$
 can be expressed as
   $ 
    \delta_i[t+1] =
     \min\left(\delta_i[t] +  
    1-\mathds{1}_{\{\beta[t]=i,\rho_2[t]=1\}}
    y_i[t], N\right)
    $, 
   and
   the evolution of the AoI of source $i$ at the relay  
 from slot $t$ to  $t+1$ can be expressed as
   $ 
    \psi_i[t+1] =
     \min\left(\psi_i[t] +  
    1-\mathds{1}_{\{\alpha[t]=i,\rho_1[t]=1\}}
    x_i[t], N\right)$. 
   We denote the state space  by  $\mathcal{S}$ which is a finite set.
     \\$\bullet $ \textit{Action}:   
      We define 
     the action   taken in slot $t$ by $\bold{a}[t]=(\alpha[t],\beta[t])$, where $\alpha[t],\beta[t]\in\{0,1,\dots,I\}$.   Let $\mathcal{A}$
  denote the action space. 
     Actions are determined by a policy, denoted by $\pi$, which is a (possibly randomized) mapping from $\mathcal{S}$ to $\mathcal{A}$. We consider \textit{stationary randomized} policies because they are dominant (see \cite[Definition 2.2]{Eitam_CMDP})
     if unichain structure\footnote{We say the unichain structure exists if  the transition probability matrix corresponding to \textit{every} stationary deterministic  policy is unichain, that is, it consists of a single recurrent class plus a possibly empty set of transient states \cite[Sec. 8.3.1]{Puterman_Book}.} exists \cite[Theorem 4.1]{Eitam_CMDP};
      we will show  in Theorem \ref{The_Unichain} below that the unichain structure exists for the transition probability matrix  of the underlying (C)MDP.
     \\$\bullet $ \textit{State Transition Probabilities}:
     We denote the state transition probability from state $\bold{s}$ to next state $\bold{s}'$ under an action $\bold{a}=(\alpha,\beta)$ by $\mathcal{P}_{\bold{s}\bold{s}'}(\bold{a})$.
    Since the evolution of the AoIs in \eqref{Eq_AoIDy}  and the arrivals are independent among the sources, the transition probability can be decomposed as
     $
     \mathcal{P}_{\bold{s}\bold{s}'}(\bold{a})=\prod_{i} 
     \Pr\{\bold{s}_i' 
     \,\big|\,\bold{s}_i,\bold{a}\}$, 
where $\Pr\{\bold{s}_i'\,\big|\,\bold{s}_i,\bold{a}\}, \forall\,i\in\mathcal{I},$ denotes the state transition probability of source $i$ under an action $\bold{a}$,  $\bold{s}_i$ is the part of the current state  associated with source $i\in\mathcal{I}$, i.e., $\bold{s}_i=(\theta_i,x_i,y_i)$, and $\bold{s}_i'$ is the part of the next state  associated with source $i$, i.e., $\bold{s}_i'=(\theta_i',x_i',y_i')$. Mathematically, $\Pr\{\bold{s}_i'\,\big|\,\bold{s}_i,\bold{a}\}$ is given by \eqref{Eq_TranPro_Unr}, shown on  top of the next page,  where 
$ \tilde{\theta}_i \triangleq \min\big(\theta_i+1, N\big)$,
 $\tilde{x}_i \triangleq  \min\big(x_i+\theta_i+1, N\big)-\min\big(\theta_i+1, N\big)$, and  
 $\tilde{y}_i \triangleq \min\big(y_i+x_i+\theta_i+1, N\big)-\min\big(x_i+\theta_i+1, N\big)$.
 \begin{theorem}\label{The_Unichain}
The transition probability matrix with elements $\mathcal{P}_{\bold{s}\bold{s}'}(\bold{a})$ corresponding to
      every  deterministic policy is unichain.
\end{theorem}
\begin{proof} 
See Appendix \ref{Proof_Unichain}.
\end{proof}
\setlength{\intextsep}{0pt}
\setlength{\textfloatsep}{0pt}
  \begin{figure*}[t]
        \begin{equation} \label{Eq_TranPro_Unr}
     \Pr\{\bold{s}_i'\,\big|\,\bold{s}_i,\bold{a}\}= \left\{
    \begin{array}{ll}
      \mu_ip_1p_2,& \alpha=i,\beta=i;~ \theta'_i=0,~x'_i=\tilde{\theta}_i,~y'_i=\tilde{x}_i,
      \\
      \mu_i(1-p_1)p_2,& \alpha=i,\beta=i;~ \theta'_i=0,~x'_i=\tilde{x}_i+\tilde{\theta}_i,~y'_i=0,
      \\
      \mu_ip_1(1-p_2), & \alpha=i,\beta=i;~ \theta'_i=0,~x'_i=\tilde{\theta}_i,~y'_i=\tilde{y}_i+\tilde{x}_i,
      \\\mu_i(1-p_1)(1-p_2), & \alpha=i,\beta=i;~ \theta'_i=0,~x'_i=\tilde{x}_i+\tilde{\theta}_i,~y'_i=\tilde{y}_i,
      \\  (1-\mu_i)p_1p_2, & \alpha=i,\beta=i;~ \theta'_i=\tilde{\theta}_i,~x'_i=0,~y'_i=\tilde{x}_i,
      \\
      (1-\mu_i)(1-p_1)p_2, & \alpha=i,\beta=i;~ \theta'_i=\tilde{\theta}_i,~x'_i=\tilde{x}_i,~y'_i=0,
      \\
      (1-\mu_i)p_1(1-p_2), & \alpha=i,\beta=i;~ \theta'_i=\tilde{\theta}_i,~x'_i=0,~y'_i=\tilde{y}_i+\tilde{x}_i,
      \\
    {(1-\mu_i)(1-p_1)(1-p_2)}, & \alpha=i,\beta=i;~ \theta'_i=\tilde{\theta}_i,~x'_i=\tilde{x}_i,~y'_i=\tilde{y}_i,
      \\
          \mu_ip_1, & \alpha=i,\beta\neq i;~ \theta'_i=0,~x'_i=\tilde{\theta}_i,~y'_i=\tilde{y}_i+\tilde{x}_i,
          \\
          \mu_i(1-p_1), & \alpha=i,\beta\neq i;~ \theta'_i=0,~x'_i=\tilde{x}_i+\tilde{\theta}_i,~y'_i=\tilde{y}_i,
      \\  (1-\mu_i)p_1, & \alpha=i,\beta\neq i;~ \theta'_i=\tilde{\theta}_i,~x'_i=0,~y'_i=\tilde{y}_i+\tilde{x}_i,
      \\
      (1-\mu_i)(1-p_1), & \alpha=i,\beta\neq i;~ \theta'_i=\tilde{\theta}_i,~x'_i=\tilde{x}_i,~y'_i=\tilde{y}_i,
  \\
          \mu_ip_2, & \alpha\neq i,\beta= i;~ \theta'_i=0,~x'_i=\tilde{x}_i+\tilde{\theta}_i,~y'_i=0,
          \\
          \mu_i(1-p_2), & \alpha\neq i,\beta= i;~ \theta'_i=0,~x'_i=\tilde{x}_i+\tilde{\theta}_i,~y'_i=\tilde{y}_i,
      \\  (1-\mu_i)p_2,& \alpha\neq i,\beta= i;~ \theta'_i=\tilde{\theta}_i,~x'_i=\tilde{x}_i,~y'_i=0,
      \\
      (1-\mu_i)(1-p_2),& \alpha\neq i,\beta= i;~ \theta'_i=\tilde{\theta}_i,~x'_i=\tilde{x}_i,~y'_i=\tilde{y}_i,
\\           \mu_i, & \alpha\neq i,\beta\neq i;~ \theta'_i=0,~x'_i=\tilde{x}_i+\tilde{\theta}_i,~y'_i=\tilde{y}_i,
      \\  1-\mu_i,& \alpha\neq i,\beta\neq i;~ \theta'_i=\tilde{\theta}_i,~x'_i=\tilde{x}_i,~y'_i=\tilde{y}_i,
      \\ 0& \text{otherwise}.
      \end{array}
      \right.
      \end{equation}
  \hrule
  \end{figure*}
   $\bullet $ \textit{Cost Functions}:
      The (immediate) cost functions include: 1) the AoI cost,  and 2) the transmission cost.
    The  AoI cost  is the \blue{weighted} sum of AoIs at the destination, i.e.,
$ C(\bold{s}[t])=\sum_i \blue{w_i} (\theta_i[t]+x_i[t]+y_i[t]). $
The transmission  cost is given by
  $D(\bold{a}[t])=\mathds{1}_{\{\alpha[t]\neq 0\}}+\mathds{1}_{\{\beta[t]\neq 0\}}.$
   \\\indent
   Given  
   a stationary randomized policy
   $\pi$,
   we denote  the WS-AAoI cost by $J(\pi)$ 
   and the average transmission cost by $\bar{D}(\pi)$, defined as follows
           \begin{align}
           \label{Eq_J}
           J(\pi)=
        \limsup_{T\rightarrow \infty}~\frac{1}{T}
       \textstyle\sum_{t=1}^{T} 
         \Bbb{E}\left\{ C(\bold{s}[t])
        \right\},
        \\\label{Eq_D}
        \bar{D}(\pi)=
    \limsup_{T\rightarrow \infty}~ \frac{1}{T}  \textstyle\sum_{t=1}^{T} 
      \Bbb{E}\left\{ D(\bold{a}[t])
       \right\}.
             \end{align}
        Note that we have  omitted the dependence on the initial state in \eqref{Eq_J} and \eqref{Eq_D} because they do not vary with the initial state, due to the  unichain structure \cite[Proposition 8.2.1]{Puterman_Book}.
By these definitions, problem \eqref{Org_P1} can  equivalently
be recast as the following CMDP problem
   \begin{align}
      \begin{array}{ll}\label{Org_P}
             \underset{\pi\in\Pi_{\mathrm{SR}}}{\text{minimize}}~~~~ &  J(\pi)
             \\
          \text{subject~to}~~~~ &
          \bar{D}(\pi)
       \le \Gamma_{\max},
      \end{array}
           \end{align}
where $\Pi_{\mathrm{SR}}$ is the set of all stationary randomized policies.
The optimal value of  the CMDP problem \eqref{Org_P} 
is denoted by $J^*$ and an optimal policy is denoted by $\pi^*$. 
\\\indent In the section below,  we turn to solve the CMDP problem \eqref{Org_P}.
\vspace{ -1 em}
\subsection{Solving the CMDP Problem \eqref{Org_P}}\label{Sec_Opt_MDP}
In order to solve the CMDP problem \eqref{Org_P}, 
we transform it into an (unconstrained) MDP 
problem using the Lagrangian relaxation method \cite{Eitam_CMDP},\cite{Ross_Optimal}. 
 The states, the actions, and the state transition probabilities of the MDP  are the same as those of the CMDP. The immediate  cost function of the MDP is defined as $L(\bold{s}[t],\bold{a}[t];\lambda)=C(\bold{s}[t])+\lambda \big(D(\bold{a}[t]) - \Gamma_{\max}\big)$, where ${\lambda}\ge0$ is a Lagrange  multiplier.
 Accordingly, 
        the MDP problem  is defined by 
     \begin{align} 
     \label{Pro_MDP}
       \underset{\pi\in \Pi_{\mathrm{SD}}}{\text{minimize}} ~~~ \mathcal{L}(\pi,{{\lambda}}),
     \end{align}
     where  
     $\mathcal{L}(\pi,{{\lambda}})
         \triangleq\limsup_{T\rightarrow \infty}~ \frac{1}{T}
       \sum_{t=1}^{T} 
     \Bbb{E}\left\{  C(\bold{s}[t])+\lambda \big(D(\bold{a}[t])-\Gamma_{\max}\big)
       \right\}
        $ 
     is  the Lagrangian and $\Pi_{\mathrm{SD}}$ is the set of all  deterministic policies; here,
     we restrict to the class of deterministic policies without loss of optimality because there always exists an optimal deterministic policy to the MDP problem \eqref{Pro_MDP} \cite[p. 370]{Puterman_Book}, which is a result of Theorem 8.4.5 in \cite{Puterman_Book} under the unichain structure shown in Theorem \ref{The_Unichain}. 
     \begin{Defi} (\textit{$\lambda$-optimal policy})\label{Def_landa_Optimal}
     A  $\lambda$-optimal policy is a solution to \eqref{Pro_MDP} and is denoted by $\pi^*_{\lambda}$. 
     \end{Defi}
By \cite[Theorem 12.8]{Eitam_CMDP}, under the Slater and two other technical conditions\footnote{ 
    The (immediate) AoI cost must be bounded below (see \cite[Eq. 11.1]{Eitam_CMDP}) and satisfy
     the moment condition \cite[Eq. 11.21]{Eitam_CMDP}. It can be  verified that these conditions hold for our problem.}, there exists a Lagrange multiplier $\lambda^*$ 
     such that
       $     J^{*}=\underset{\pi\in \Pi_{\mathrm{SD}}} 
             {\min} ~\mathcal{L}(\pi,{{\lambda^*}}).
         $
    Moreover, if $ \bar{D}(\pi^*_{\lambda^*})=\Gamma_{\max}$, then $\pi^*_{\lambda^*}$ is  an optimal policy for  the CMDP problem \eqref{Org_P}.
 \begin{remark}\label{Remark_2_EOP}
 Under the unichain structure shown in Theorem \ref{The_Unichain}, 
 by results of \cite{Ross_Optimal}, it can be shown that 
 an optimal policy of the CMDP problem \eqref{Org_P} is  a  stationary randomized policy  
 that performs randomization 
 between two deterministic $\tilde{\lambda}$-optimal policies where one is feasible and the other is infeasible to \eqref{Org_P} (e.g., \cite{Elif_ARQ, Brancu+2Hop}).
  However, given  such  $\tilde{\lambda}$-optimal policies, finding a randomization factor of such optimal   policy is computationally difficult  \cite[Sec. 3.2]{Adam_Estimation}.
\end{remark} 
\vspace{- 0.5 em}
As stated in Remark \ref{Remark_2_EOP}, it is difficult to obtain an optimal stationary randomized policy (by randomizing the two $\tilde{\lambda}$-optimal policies) to the CMDP problem \eqref{Org_P}. Therefore, we will next develop a \textit{practical}  \textit{ near-optimal} (as empirically shown in Section \ref{Sec_Numerical_Res}) deterministic policy to the CMDP problem \eqref{Org_P}.
In particular, we propose a solution relying on bisection search over the Lagrange multiplier $\lambda$ and relative value iteration algorithm (RVIA). Namely, we alternate between solving
the MDP problem \eqref{Pro_MDP} for a given $\lambda$ and searching for  a particular value of  $\lambda$ for which $\pi^*_{\lambda}$
is feasible for problem \eqref{Org_P} and gives the best performance among all feasible $\lambda$-optimal policies. 
\subsubsection{Solution of the MDP  Problem \eqref{Pro_MDP}} Towards solving the MDP problem \eqref{Pro_MDP}, first,
we present the following theorems
related to a $\lambda$-optimal policy; 
particularly, Theorem \ref{The_Bellman} characterizes a ${\lambda\text{-optimal}}$ policy  and Theorem \ref{Th_Detr} specifies its structure.
Then, we utilize these theorems to develop a structure-aware RVIA \cite[Sec. 8.5.5]{Puterman_Book} that gives a $\lambda$-optimal policy.
\begin{theorem}\label{The_Bellman}
There exists $h(\bold{s})$, for each state  $\bold{s} \in \mathcal{S}$, such that  
\begin{align}\label{Eq_Bellman}
    \begin{array}{cc}
        \mathcal{L}^{*}(\lambda) + h(\bold{s}) = \min_{\bold{a}\in\mathcal{A}}
        \{L(\bold{s},\bold{a};\lambda)+{\textstyle\sum_{\bold{s}'\in\mathcal{S}}}\mathcal{P}_{\bold{s}\bold{s}'}(\bold{a})h(\bold{s}')\},
    \end{array}
\end{align}
where $\mathcal{L}^{*}(\lambda)$ is the optimal value of the MDP problem \eqref{Pro_MDP} for given $\lambda$. Moreover, an optimal action taken by a $\lambda$-optimal policy in each state $\bold{s} \in \mathcal{S}$, $\pi^*_{\lambda}(\bold{s})$, is given by
\begin{equation}\label{Eq_Optimal_policy}
\begin{array}{cc}
  \pi^*_{\lambda}(\bold{s}) \in \arg\min_{\bold{a}\in\mathcal{A}}
        \left\{L(\bold{s},\bold{a};\lambda)+{\textstyle\sum_{\bold{s}'\in\mathcal{S}}}\mathcal{P}_{\bold{s}\bold{s}'}(\bold{a})h(\bold{s}')\right\},~ \bold{s}\in\mathcal{S}. 
\end{array}
\end{equation}
\end{theorem}
\begin{proof}
Because of  Theorem \ref{The_Unichain}, the first part, i.e., \eqref{Eq_Bellman}, follows from \cite[Theorem 8.4.3]{Puterman_Book}. Then, the second part, i.e., \eqref{Eq_Optimal_policy}, directly follows from \cite[Theorem 8.4.4]{Puterman_Book}.
\end{proof}
 \begin{theorem}
   \label{Th_Detr} 
 Any $\lambda$-optimal policy of problem \eqref{Pro_MDP}
 has
     a switching-type structure
     for $\beta$ with respect to $\bold{y}=(y_1,\dots,y_I)$. 
     This is, if the policy takes action $\beta=i,~i\in\{1,\dots,I\},$ at state $\bold{s}$, it also
      takes the same action at all states  $\bold{s}+k\bold{e}_{3i}$, for all $k\in\Bbb{N}$, where 
     $\bold{e}_{3i}$ is a vector  in which the $(3i)$-th element is $1$ and the others are $0$.
 \end{theorem}
 \begin{proof}
 See Appendix \ref{Proof_App_Threshod}.
\end{proof}
RVIA is an iterative procedure that utilizes the optimality equation \eqref{Eq_Bellman}. Particularly, at each iteration $n\in\{0,1,\dots\}$, for each state $\bold{s} \in\mathcal{S}$, we have 
\begin{equation}
\label{Eq_RVI}
\begin{array}{ll}
    & V_{n+1}(\bold{s})=\min_{\bold{a}\in\mathcal{A}}\left\{ L(\bold{s},\bold{a};\lambda)+{\textstyle\sum_{\bold{s}'\in\mathcal{S}}}\mathcal{P}_{\bold{s}\bold{s}'}(\bold{a})h_n(\bold{s}')\right\},~~ 
    \\
    & 
    h_{n+1}(\bold{s}) = V_{n+1}(\bold{s})-V_{n+1}(\mathbf{s}_\mathrm{ref}),
\end{array}
\end{equation}
\sloppy 
where 
$\mathbf{s}_\mathrm{ref} \in \mathcal{S}$ is an arbitrarily chosen reference state. 
The   structure-aware RVIA is presented in Alg. \ref{A_RVI} (see Steps 3--16), where $\varepsilon$ is a small constant for the RVIA termination criterion. In particular, at each iteration of RVIA, in Steps 6--9, we exploit the switching-type  structure  specified in Theorem \ref{Th_Detr}  to find an optimal action for each state $\bold{s}$, i.e., ${\bold{a}^* \triangleq {\arg\min_{\bold{a}\in\mathcal{A}}} 
 \{L(\bold{s},\bold{a};\lambda)+ 
 {\sum_{\bold{s}'\in\mathcal{S}}}\mathcal{P}_{\bold{s}\bold{s}'}(\bold{a})h(\bold{s}') \}} $;
specifically,  in computing ${\bold{a}^*}\triangleq (\alpha^*,\beta^*)$, whenever we have determined an optimal  decision of $\beta^*$ in Step 6, then we only need to find an optimal decision of  $\alpha^*$. 
\\\indent
The following theorem shows that RVIA given by \eqref{Eq_RVI}  (i.e., Steps 3--16 of Alg. \ref{A_RVI}) converges and  returns the optimal value of the MDP problem \eqref{Pro_MDP}.
\vspace{-0.5 em}
\begin{theorem}\label{The_RVIA_Converg}
For any initialization $V_0(\bold{s})$, the sequences $\{h_n(\bold{s})\}_{n=1,2,\dots}$  and $\{V_n(\bold{s})\}_{n=1,2,\dots}$, generated by \eqref{Eq_RVI},  converge, 
i.e., ${\lim_{n\rightarrow \infty}~h_{n}(\bold{s}) \triangleq h(\bold{s}) }$ and ${\lim_{n\rightarrow \infty}~V_{n}(\bold{s}) \triangleq V(\bold{s})}$. Moreover,  ${h(\bold{s}) = V(\bold{s}) - V(\bold{s}_\mathrm{ref})} $ satisfies \eqref{Eq_Bellman} and
$V(\bold{s}_\mathrm{ref}) = \mathcal{L}^{*}(\lambda)$.
\end{theorem}
\begin{proof}
The proof  follows directly from \cite[Prop. 4.3.2]{bertsekas2007dynamic}. Thus, we need to show that the hypothesis of \cite[Prop. 4.3.2]{bertsekas2007dynamic} holds.
According to \cite[p. 209]{bertsekas2007dynamic}, it is sufficient to show that  
the  Markov chain, described by the transition probability matrix with elements $\mathcal{P}_{\bold{s}\bold{s}'}(\bold{a})$, corresponding to every deterministic policy, is unichain and aperiodic.
The unichain structure has been proven in Theorem \ref{The_Unichain},
and aperiodicity follows from the fact that 
the recurrent state $\bold{s}^{\mathrm{acc}}$ (see Appendix \ref{Proof_Unichain}) has self transition (i.e., ${\mathcal{P}_{\bold{s}^{\mathrm{acc}}\bold{s}^{\mathrm{acc}}}(\bold{a}) >0,~\forall\, \bold{a}\in\mathcal{A}}$); this is because from  \cite[Exercise 4.1]{Gallager_SP_Theory_App}, such (recurrent) state is also aperiodic, and  then, by \cite[Theorem 4.2.8]{Gallager_SP_Theory_App}, 
all states that belong to the same class as the recurrent state $\bold{s}$ are aperiodic.
\end{proof}
\setlength{\textfloatsep}{0pt}
\begin{algorithm} [t!]
\SetAlgoLined
\setlength{\AlCapSkip}{1em}
 \KwInput{$
 1)~ \text{System parameters:}~\Gamma_{\max},~I,~\{\mu_i,w_i\}_{i\in\mathcal{I}},~
  p_1,~p_2,$
  \\
  $2)~ \text{RVIA and bisection parameters:}~
  N,~\zeta,~\varepsilon,~\lambda^{+},~\lambda^{-},
  ~\text{and}~ 3)~\text{arbitrarily chosen}~ \bold{s}_{\mathrm{ref}}\in\mathcal{S}$ }   
\tcp{Bisection search over $\lambda$}
 \While{ $|\lambda^{+}-\lambda^{-}|\ge\zeta$ }
 {
 
 $\lambda_{\mathrm{bis}}=\frac{\lambda^{+}+\lambda^{-}}{2}$; 
 
 \tcp{Initialization of RVIA}
 $\text{Set for each}  ~\bold{s}\in\mathcal{S}:~V(\bold{s})=0,~h(\bold{s})=0,~h_{\mathrm{old}}(\bold{s})=1$\;  
 
 \tcp{RVIA for a given $\lambda_{\mathrm{bis}}$}
 \While{$
 \max_{\bold{s}\in\mathcal{S}}|h(\bold{s})-h_{\mathrm{old}}(\bold{s})|>\varepsilon$} 
 {
 
 \For   {$\text{each}~\bold{s}\in\mathcal{S}$}{
 \tcp{Using the switching-type structure}
  \eIf {there exists $k\in\Bbb{N}$  such that  $\beta=i$ for $\bold{s}-k\bold{e}_{3i}$}{
 
   $\bold{a}^{*}\leftarrow(\alpha^{*},i)$, where
 
 $\alpha^{*}=\arg\min_{\alpha\in\{0,1,2\}}\{L(\bold{s},\bold{a};\lambda_{\mathrm{bis}})+
 {\textstyle\sum_{\bold{s}'\in\mathcal{S}}}\mathcal{P}_{\bold{s}\bold{s}'}(\bold{a})h(\bold{s}')
 \} $\;
   }
   {
 $
 \bold{a}^{*}\leftarrow 
 {\arg\min_{\bold{a}\in\mathcal{A}}} 
 \{L(\bold{s},\bold{a};\lambda_{\mathrm{bis}})+ {\textstyle\sum_{\bold{s}'\in\mathcal{S}}}\mathcal{P}_{\bold{s}\bold{s}'}(\bold{a})h(\bold{s}') \} $\;
 }
 
 $
 V(\bold{s})\leftarrow L(\bold{s},\bold{a}^{*};\lambda_{\mathrm{bis}})+{\textstyle\sum_{\bold{s}'\in\mathcal{S}}}\mathcal{P}_{\bold{s}\bold{s}'}(\bold{a})h(\bold{s}')$\; 
 
  $h_{\mathrm{tmp}}(\bold{s})\leftarrow V(\bold{s})-V(\mathbf{s}_\mathrm{ref})$\; 
  
  }
  $h_{\text{old}}(\bold{s})\leftarrow h(\bold{s})$, $h(\bold{s})\leftarrow h_{\mathrm{tmp}}(\bold{s})$\;
 
 }

  Compute $\bar{D}(\pi^*_{\lambda_{\mathrm{bis}}})$\; 
 \eIf {$\bar{D}(\pi^*_{\lambda_{\mathrm{bis}}})> \Gamma_{\max}$}{
 
  $ \lambda^{-}\leftarrow\lambda_{\mathrm{bis}}$\;
  
  }
   { $\lambda^{+}\leftarrow \lambda_{\mathrm{bis}}$\;
  
   }
 }

 Compute  $\pi^*_{\lambda^{+}}$ and $\pi^*_{\lambda^{-}}$ using \eqref{Eq_Optimal_policy}\;
 
  \KwOutput{$\pi^*_{\lambda^-}$ and $\pi^*_{\lambda^+}$}
 \caption{\small 
Structure-aware RVIA with bisection  to solve  the CMDP problem \eqref{Org_P} 
} \label{A_RVI}
\end{algorithm} 
\subsubsection{Searching for Lagrange Multiplier}
\textcolor{blue}{
By  \cite[Lemma 3.1]{Ross_Optimal},
 $J(\pi^{*}_{\lambda})$ is increasing  in $\lambda$ and $\bar{D}(\pi^{*}_{\lambda})$ is decreasing in $\lambda$\footnote{
 \textcolor{blue}{
 Intuitively, 
 increasing $\lambda$ penalizes more the transmission cost in the Lagrangian; thus, by increasing $\lambda$, the average number of transmissions $\bar{D}(\pi^{*}_{\lambda})$ decreases, which, in turn, increases the WS-AAoI $J(\pi^{*}_{\lambda})$.}
 }.
 } 
Accordingly, we are interested in the smallest value of Lagrange multiplier $ \lambda$ for which policy $\pi^{*}_{\lambda}$ is feasible for the CMDP problem \eqref{Org_P}.
Formally, we want to find ${\tilde{\lambda}\triangleq\inf\{\lambda>0: \bar{D}(\pi^{*}_{\lambda})\le \Gamma_{\max}\}}$. 
To search for  $\tilde{\lambda}$, we apply  bisection
 that exploits the monotonicity of  $\bar{D}(\pi^{*}_{\lambda})$ with respect to $\lambda$.
We initialize bisection with  $\lambda^{-}=0$ and $\lambda^{+}$ as a large positive real number. Then, bisection iterates until $|\lambda^{+}-\lambda^{-}|\le \zeta$, where  
 $\zeta$ is a 
small constant for the bisection termination criterion. 
 Details are stated in Alg. \ref{A_RVI}. 
\\\indent
It is worth stressing  that, as stated in Remark \ref{Remark_2_EOP}, 
 there is no guarantee, 
 even for an arbitrarily small  $\zeta$, 
that  the feasible deterministic policy $\pi^*_{\lambda^+}$, obtained by Alg. \ref{A_RVI},
would be an optimal policy  for the  CMDP problem \eqref{Org_P}. Nevertheless, the empirical results in Section \ref{Sec_Numerical_Res} will show that policy $\pi^*_{\lambda^+}$ has near-optimal performance. 
At the same time, the infeasible policy $\pi^*_{\lambda^-}$
can serve as a benchmark, because it provides a lower bound to an optimal solution of 
\eqref{Org_P}. 
In Section \ref{Sec_Numerical_Res}, we will empirically show that
policy $\pi^{*}_{\lambda^-}$ is a tight lower bound solution.
\\\indent
It is essential to note that the computational complexity of the (relative) value iteration algorithms dramatically grows as the state and action spaces increase, i.e., the \textit{curse of dimensionality problem};
\textcolor{blue}{
the detailed complexity analysis of Alg. \ref{A_RVI} can be found in Sec. \ref{Sec_ComAna}.}
Since RVIA is run at each iteration of  bisection,
Alg.~\ref{A_RVI} 
 becomes computationally inefficient  
when  applied 
for 
a large number of sources.
To circumvent the curse of dimensionality, we propose a low-complexity scheduling policy  in the next section.
 \vspace{- 1 em}
\section{Low-Complexity Scheduling Policy to Solve Problem \eqref{Org_P1}}\label{Sec_LC_MW}
In this section, we devise   DPP-SP (i.e., drift-plus-penalty-based scheduling policy), using the idea of the drift-plus-penalty method \cite{Neely_Sch},  to solve the main problem \eqref{Org_P1}.
The proposed  DPP-SP is a heuristic policy that has low complexity and, as empirically shown in Section  \ref{Sec_Numerical_Res}, obtains a near-optimal performance.
We prove that DPP-SP is guaranteed to satisfy constraint \eqref{Con_Sam1}. 
\\\indent
According to the drift-plus-penalty method \cite{Neely_Sch}, the time average constraint   \eqref{Con_Sam1} is enforced by transforming it into queue stability constraint. 
Accordingly, a virtual queue is associated  for  constraint \eqref{Con_Sam1}   in such a way  that the stability of the virtual queue implies satisfaction of the constraint.
Let $H[t]$ denote the virtual queue associated with  constraint \eqref{Con_Sam1}  in slot $t$  which evolves as
\vspace{-1 em}
\begin{align}\label{Eq_Evol_Q}
\begin{array}{cc}
    H[t+1]=\max\{H[t]-\Gamma_{\max}+D(\bold{a}[t]),~0\}. 
\end{array}
\end{align}
By \cite[Ch. 2]{Neely_Sch}, the time average constraint \eqref{Con_Sam1} is satisfied if
 the  virtual queue is \textit{strongly stable}, i.e.,
         $\limsup_{T\rightarrow \infty} ~\frac{1}{T} \sum_{t=1}^{T}\Bbb{E}\{H[t]\} <+\infty$.
Next, we define the Lyapunov function and its drift which are used to define the virtual queue stability condition.
\\\indent
We  define a quadratic  Lyapunov function as  $L(H[t])=\frac{1}{2}H^2[t]$  \cite[Ch. 3]{Neely_Sch}.
The Lyapunov function indicates the size of the virtual queue, i.e., if
the Lyapunov function is small, then  the virtual queue is small,
and if the Lyapunov function is large, then the virtual queue is large. By minimizing the expected change of
the Lyapunov function from one slot to the next, the virtual queue can be stabilized \cite[Ch. 4]{Neely_Sch}. Let $\mathcal{Z}[t]=\{\bold{s}[t],H[t]\}$ denote the system state in slot $t$. The one-slot \textit{conditional Lyapunov drift},
denoted by $\Delta[t]$, is defined as the expected change in the Lyapunov function over one slot given the current system state $\mathcal{Z}[t]$. 
Accordingly,  $\Delta[t]$ is given by
\vspace{-1.5  em}
\begin{equation}\label{Eq_Drift_F}
    \begin{array}{ll}
        \Delta[t]=\Bbb{E}\left\{L(H[t+1])-L(H[t])~\big|~\mathcal{Z}[t]\right\},
    \end{array}
\end{equation}
 where the expectation is with respect to the
(possibly random)  decisions made in reaction to the current system state.
 \\\indent
Applying the drift-plus-penalty method to main problem \eqref{Org_P1}, we seek for a control policy that minimizes an upper bound on the following drift-plus-penalty  function, $\varphi[t]$, at every  slot $t$: 
     \begin{equation}\label{Eq_DPP_Or_1}
     \begin{array}{ll}
      \varphi[t] & = \Delta[t]+ V 
      \textstyle{\sum_i}w_i \Bbb{E}\{  (\delta_i[t+1] + \psi_i[t+1])~\big|~\mathcal{Z}[t]\} 
      \\ &= \Delta[t]+V \sum_i 
       w_i
      \Bbb{E}\left\{
      \left(2\theta_i[t+1]+2x_i[t+1]+y_i[t+1]\right)~\big|~\mathcal{Z}[t]\right\},  
     \end{array}
 \end{equation}
     where the expectation is with respect to the channel randomness (i.e., $\rho_1[t]$ and $\rho_2[t]$) and
(possibly random)  decisions made in reaction to the current system state; parameter $V \ge 0$ adjusts a trade-off between  the size of the virtual queue   and the objective function.
It is noteworthy  that, in \eqref{Eq_DPP_Or_1}, different from considering the original immediate objective function (i.e., the sum AoI at the destination) as the penalty term, we have added the sum AoI at the relay to the penalty term so that  
minimizing the upper bound of the drift-plus-penalty function at each slot also concerns  the evolution of the sum AoI at the relay. 
\\\indent
To obtain the upper bound of the drift-plus-penalty function, we derive an upper bound for the drift term $\Delta[t]$, given by the following proposition. 
 \begin{Pro}\label{Lemma_Upper_bound}
The upper bound for the conditional
Lyapunov drift in \eqref{Eq_Drift_F} is given by 
\begin{equation}
    \begin{array}{ll}\label{Eq_Drift_1}
         \Delta[t] 
         \le B + H[t]\big(\Bbb{E}\{D(\bold{a}[t])~\big|~\mathcal{Z}[t]\}-\Gamma_{\max}\big),
    \end{array}
\end{equation}
where $B=1/2\Gamma_{\max}^2+2$.
 \end{Pro}
 \begin{proof}
   See Appendix \ref{Appendix_UpperB}.
 \end{proof}
 
Let us express the evolution of the AoI and the relative AoIs of each source $i\in\mathcal{I}$ by  the following compact formulas\footnote{These expressions are for unbounded AoI values as the derivation of DPP-SP does not require to bound them.}
\begin{equation}\label{Eq_AoIs_Compact}
    \begin{array}{ll} 
      & \theta_i[t+1]=(1-u_i [t+1])(\theta_i[t]+1),~
      \\ & x_i[t+1]=(1-\rho_1[t]\mathds{1}_{\{\alpha[t]=i\}})x_i[t]+u_i [t+1](\theta_i[t]+1), 
     \\& y_i[t+1]=(1-\rho_2[t]\mathds{1}_{\{\beta[t]=i\}})y_i[t]+\rho_1[t]\mathds{1}_{\{\alpha[t]=i\}}x_i[t].
    \end{array}
    \end{equation}
   Using Proposition \ref{Lemma_Upper_bound} and substituting \eqref{Eq_AoIs_Compact}  into   \eqref{Eq_DPP_Or_1},  the upper bound for the drift-plus-penalty function $\varphi[t]$ can be derived as
\begin{equation}\label{Eq_upperB_DPP}
    \begin{array}{ll}
         \varphi[t]
         &
         \le
         B +  H[t](\Bbb{E}\{D(\bold{a}[t])~\big|~\mathcal{Z}[t]\}-\Gamma_{\max})
         \\&
         + V \sum_i w_i
         \left(\Bbb{E}\left\{(1-\rho_2[t]\mathds{1}_{\{\beta[t]=i\}})y_i[t]+(1-\rho_1[t]\mathds{1}_{\{\alpha[t]=i\}})x_i[t] + x_i[t]+2\theta_i[t]+2~\big|~\mathcal{Z}[t]\right\}\right).
    \end{array}
\end{equation}
\indent
Now, we turn to minimize the  upper bound of the drift-penalty-function given in \eqref{Eq_upperB_DPP}.
To this end, we first compute the expectations with respect to the channel randomness, i.e., we have ${\Bbb{E}\{\rho_2[t]\mathds{1}_{\{\beta[t]=i\}}~|~\mathcal{Z}[t]\}=p_2\Bbb{E}\{\mathds{1}_{\{\beta[t]= i\}}~|~\mathcal{Z}[t]\}}$ and $ {\Bbb{E}\{\rho_1[t]\mathds{1}_{\{\alpha[t]=i\}}~|~\mathcal{Z}[t]\}=p_1\Bbb{E}\{\mathds{1}_{\{\alpha[t]= i\}}~|~\mathcal{Z}[t]\}}$. Then, after removing  the terms  in \eqref{Eq_upperB_DPP}
that are  independent of the decision variables,
we need to minimize the following expression:
\begin{equation}\label{Eq_MW_UB_Simplified}
\begin{array}{ll}
&
H[t]\Bbb{E}\{\mathds{1}_{\{\beta[t]\neq 0\}}~|~\mathcal{Z}[t]\}
   -Vp_2\sum_i w_i \Bbb{E}\{\mathds{1}_{\{\beta[t]= i\}}~|~\mathcal{Z}[t]\} y_i[t] 
   \\ 
  &
   +  H[t]\Bbb{E}\{\mathds{1}_{\{\alpha[t]\neq 0\}}~|~\mathcal{Z}[t]\}
   -Vp_1\sum_i w_i
   \Bbb{E}\{\mathds{1}_{\{\alpha[t]= i\}}~|~\mathcal{Z}[t]\} x_i[t], 
    \end{array}
\end{equation}
where the expectation is with respect to the (possibly random) decisions.
\\\indent
 To minimize the expression in \eqref{Eq_MW_UB_Simplified}, we follow the approach of opportunistically minimizing a (conditional) expectation
\cite[p. 13]{Neely_Sch}, i.e., the expression in \eqref{Eq_MW_UB_Simplified} is minimized by the algorithm that observes the current system state $\mathcal{Z}[t]$ and chooses  $\alpha[t]$ and $\beta[t]$ to minimize 
\begin{equation}\label{Eq_dpp_sp_2}
    \begin{array}{cc}
        H[t]\mathds{1}_{\{\alpha[t]\neq 0\}}
  -Vp_1\sum_i w_i \mathds{1}_{\{\alpha[t]= i\}} x_i[t]
  +
      H[t]\mathds{1}_{\{\beta[t]\neq 0\}}
  -Vp_2\sum_i w_i \mathds{1}_{\{\beta[t]= i\}} y_i[t].
    \end{array}
\end{equation}
The expression in  \eqref{Eq_dpp_sp_2} is separable with respect to   $\alpha[t]$ and $\beta[t]$,  
 thus we obtain $\alpha[t]$ and $\beta[t]$ by solving the following problems 
 \vspace{-1 em}
\begin{align}\label{Eq_alpha}
      \underset{\alpha[t]\in\{0,1,\dots,I\}}{\text{minimize}}~~
      H[t]\mathds{1}_{\{\alpha[t]\neq 0\}}
  -Vp_1 \textstyle\sum_i w_i \mathds{1}_{\{\alpha[t]= i\}} x_i[t], 
    \\
    \label{Eq_beta}
      \underset{\beta[t]\in\{0,1,\dots,I\}}{\text{minimize}}~~
      H[t]\mathds{1}_{\{\beta[t]\neq 0\}}
  -Vp_2 \textstyle \sum_i w_i \mathds{1}_{\{\beta[t]= i\}} y_i[t]. 
\end{align}
 It can be inferred from   problem \eqref{Eq_alpha} that if $H[t]\ge \max_{i\in\mathcal{I}} \{Vp_1 w_i x_i[t]\}$, then the optimal action is $\alpha[t]=0$; otherwise, the optimal action is $\alpha[t]=\argmax_{i\in\mathcal{I}} \{Vp_1 w_ix_i[t]\}$. Problem \eqref{Eq_beta} has the similar solution with respect to $\beta[t]$.
\\\indent
In summary,  the proposed DPP-SP works as follows:
 at each slot $t$, the controller observes $\mathcal{Z}[t]$\vspace{-0.5em}\footnote{\textcolor{blue}{It is worthwhile to mention that the observed state $\mathcal{Z}[t]$ has 
 the virtual queue $H[t]$ in addition to  
the observed state of RVIA. However, the virtual queue
is just an extra variable maintained in the
internal memory and updated according to the action taken in the past
and consists of only one variable regardless of the number of sources.
Thus, observing the virtual queue
does not need any signaling or exchange of information in the system because it is, as the name
suggests, virtually created by the controller and its dynamic only depends on its current value and the transmission actions taken.
}
}
and
 determines  the transmission decision variables according to the following rules
 \begin{equation}\label{Eq_MW_Policy_1}
  \begin{array}{ll}
      \text{If} ~\max_{i\in\mathcal{I}} \{Vp_1  w_i x_i[t]\} \ge
      H[t],~ \text{then}~ \alpha[t]=\argmax_{i \in\mathcal{I}} \{V p_1  w_i x_i[t]\};~ \text{otherwise},~ \alpha[t]=0,
\\
\text{If} ~\max_{i\in\mathcal{I}} \{V p_2 w_i y_i[t]\} \ge H[t], ~ \text{then}~ \beta[t]=\argmax_{i\in\mathcal{I}} \{V p_2 w_i y_i[t]\}; ~ \text{otherwise},~ \beta[t]=0.
  \end{array}  
\end{equation}
\indent What remains is to show that DPP-SP, operating  according to  \eqref{Eq_MW_Policy_1}, satisfies constraint \eqref{Con_Sam1}.
We prove this in the following theorem.
 \begin{theorem}\label{Th_Stability}
 Assume that $\Bbb{E}\{L(H[0])\}$ is finite. For any finite $V$,
  the  virtual queue under  DPP-SP that operates according to  \eqref{Eq_MW_Policy_1} is strongly stable, implying that DPP-SP satisfies constraint \eqref{Con_Sam1}.
 \end{theorem}
 \begin{proof}
 See Appendix \ref{App_Stability}.
 \end{proof}
 As it can be seen in \eqref{Eq_MW_Policy_1},  DPP-SP performs only \textit{two simple operations} to determine the actions at each slot.
Hence,  DPP-SP has low complexity  and can easily support systems with large numbers of sources.
\textcolor{blue}{
The detailed complexity analysis of DPP-SP can be found in Sec. \ref{Sec_ComAna}.}

\vspace{-1 em}
\section{ A Deep Reinforcement Learning Algorithm to Solve Problem \eqref{Org_P1}}\label{Sec_Learning}
In this section,   we develop a   deep reinforcement learning algorithm to solve the main problem \eqref{Org_P1}.
Inspired by \cite{Lya_DRL_Indus}, we use the Lyapunov optimization theory to convert the CMDP problem \eqref{Org_P}  into an
MDP problem which is then solved by a model-free deep learning algorithm, namely, D3QN (i.e., dueling double deep Q-network) \cite{D3DQN_First,D3QN_II_TWC}.
\textcolor{blue}{
Note that another approach to the CMDP problem \eqref{Org_P} could be a primal-dual reinforcement learning algorithm.  In contrast to our algorithm, such an algorithm  
leads to an iterative optimization procedure. Thus, the proposed Lyapunov-based learning algorithm is in general simpler than a primal-dual DRL-based algorithm.
}

It is worth pointing that: i) as D3QN is a model-free algorithm, we do not require 
the state transition probabilities of the MDP  problem, thus, the proposed deep learning  is applicable  for unknown environments (i.e., when  the packet arrival rates and the successful transmission probabilities of the (wireless) links  are not available at the controller), 
and ii) there is no guarantee that the proposed deep learning algorithm provides an optimal policy to the main problem \eqref{Org_P1};
however, an advantage of the deep learning algorithm  is coping with unknown environments with large state and/or action spaces which can be used as a benchmark policy.
We further note that to implement the proposed learning algorithm, we do not need to bound the AoI values and store the state space (which may require considerable memory).
\\\indent
We define the  expected time average reward function, obtained by policy $\pi$, as
\begin{align}
\begin{array}{cc}
     R(\pi)\triangleq
   \limsup_{T\rightarrow \infty} \frac{1}{T}
      \textstyle \sum_{t=0}^{T}
 \Bbb{E}\left\{ r[t]\right\},  
\end{array}
\end{align}
where  
$ 
r[t] = 
-\Big(L(H[t+1])-L(H[t])+V \sum_i \textcolor{blue}{w_i} \delta_i[t+1] \Big)
$ 
is the \textit{immediate reward function}, and
$L(H[t])=\frac{1}{2}H^2[t]$ is  the quadratic  Lyapunov function with virtual queue $H[t]$  given by \eqref{Eq_Evol_Q}.
\blue{It is worth pointing out that the Lyapunov drift in the reward function  is introduced  to guarantee the satisfaction of the average constraint \eqref{Con_Sam1} \cite{Lya_DRL_Indus, TWC_Learning}. }
We want to solve the following problem
\begin{equation}\label{Prob_LR}
    \begin{array}{ll}
         \underset{\pi}{\text{maximize}}~~R(\pi).
    \end{array}
\end{equation} 
\indent
Problem \eqref{Prob_LR} 
can be formulated as  an MDP problem, where $r[t]$ is the immediate reward, the state is ${\mathcal{Z}[t]=\{\bold{s}[t],H[t]\}}$, and the action is ${\bold{a}[t]=(\alpha[t],\beta[t])}$. To solve the MDP problem,  we apply D3QN. Implementation details are presented in Sec. \ref{Sec_Numerical_Res}.
\vspace{-1 em}
\textcolor{blue}{
\section{Complexity Analysis}\label{Sec_ComAna}
Here, we analyze the (overall) \textit{computational} complexity of the proposed policies.
First,  in terms of complexity, there are two different phases: 1) offline phase, i.e., an initial phase to find a policy, and 2) online phase where the (offline-derived) policy is used to generate the corresponding action at each slot. DPP-SP does not have the offline phase, whereas the deterministic policy obtained by Alg. \ref{A_RVI} and the deep learning policy have both the offline and online phases. 
Next, we elaborate the complexity of the proposed polices in each phase. The complexity of the policies is summarized in Table \ref{Table_Com}.
\\\indent
$\bullet$ \underline{The deterministic policy:} 
The complexity of the offline phase of the deterministic policy is the complexity of running Alg. \ref{A_RVI}. 
Alg. \ref{A_RVI} is an iterative algorithm that involves iterating between bisection and  RVIA.  
The complexity order of   \text{each iteration} of RVIA is at most $\mathcal{O}(|\mathcal{A}||\mathcal{S}|^2)$, where 
 the state space size $|\mathcal{S}|$ is approximately $N^{3I}$ and the action space size $|\mathcal{A}|$ is $(I+1)^2$.
Accordingly, the complexity of the offline phase  of the deterministic  policy is $ {  \mathcal{O}\left( M_1 M_2 I^2 N^{6I} \right) }$, where $M_1$ and $M_2$, are, respectively, the iterations required in bisection and RVIA.  
The complexity of the online phase is $\mathcal{O}(1)$ since it is needed to just fetch the corresponding action of each state from the  lookup table  obtained in the offline phase.
\\\indent
$\bullet$ \underline{DPP-SP:} 
As mentioned above, DPP-SP does not have the offline phase. In the online phase, the policy needs $I$ comparisons for each of the two decision variables, thus, $2I$ comparisons in total. Therefore, the complexity of DPP-SP in the online phase is $\mathcal{O}(I)$.
\\\indent
$\bullet$ 
\underline{The deep learning policy:} 
The offline phase of the deep learning policy is its training phase. 
Because the policy is based on the deep neural network, its  (computational) complexity is mainly related to the model and size of the neural network and the training process.
The training complexity of the neural network consists of two stages: 1) the forward propagation   algorithm (forward pass) and 2) the backpropagation  algorithm (backward pass).
The complexity of the forward propagation  algorithm is $\mathcal{O}\left( P_\mathrm{h}(P_\mathrm{i}+ M_3P_\mathrm{h} + P_\mathrm{o})\right)$ \cite{Multi_Agent_WCOML,Optimal_VNF}, where
$P_\mathrm{i} = 1+3I$ is the number of neurons of the input layer (which equals the number of elements in the state vector),  and $P_\mathrm{o} = |\mathcal{A}| $ is the number of neurons of the output layer. Moreover, 
$P_\mathrm{h}$ is the number of neurons in each hidden layer and $M_3$ is the number of hidden layers;  it is assumed that all the hidden layers have the same number of neurons.
The complexity of the backpropagation algorithm is similar to that of the forward propagation algorithm \cite{Multi_Agent_WCOML}. 
In terms of the training process, the complexity is mainly related to the number of episodes $M_4$ and iterations (per episode) $M_5$, and the batch size $M_6$ (i.e., the number of samples used to update the weights of the neural network). 
Accordingly, the overall complexity of the offline phase of the policy is $\mathcal{O}\left( M_4M_5M_6P_\mathrm{h}(P_\mathrm{i}+ M_3P_\mathrm{h} + P_\mathrm{o}) \right)$.
In the online phase, the action selection is done by executing the forward propagation algorithm and thus, the complexity of the online phase of the policy is  $\mathcal{O}( P_\mathrm{h}(P_\mathrm{i}+ M_3P_\mathrm{h} + P_\mathrm{o}))$.
}
\begin{table}[t!]
\centering
\caption{ \textcolor{blue}{The overall computational complexity of the proposed policies } }
\label{-1 em }
	\begin{tabular}{ c |  c | c}
		\hline
		\textbf{Policy}  & \textbf{Offline phase} & \textbf{ Online phase}
   \\
		\hline
		\text{Deterministic policy} & $\mathcal{O}\left( M_1 M_2 I^2 N^{6I} \right)$  & $\mathcal{O}(1)$
		\\
		\hline
	\text{DPP-SP}  & --- &  $\mathcal{O}(I)$
 		\\
		\hline
	\text{Deep learning policy}  &  
 $\mathcal{O}\left( M_4M_5M_6P_\mathrm{h}(P_\mathrm{i}+ M_3P_\mathrm{h} + P_\mathrm{o}) \right)$   &
$\mathcal{O}\left( P_\mathrm{h}(P_\mathrm{i}+ M_3P_\mathrm{h} + P_\mathrm{o}) \right)$
\\\hline
	\end{tabular}
 \label{Table_Com}
\end{table}


\vspace{-1 em}
\section{Numerical Results}\label{Sec_Numerical_Res}
In this section, we numerically evaluate  the WS-AAoI (i.e., weighted  sum average AoI at the destination) performance of the three proposed  policies: 1) the deterministic policy $\pi^*_{\lambda^+}$ obtained by the structure-aware  RVIA in Alg. \ref{A_RVI}, 2) DPP-SP given by \eqref{Eq_MW_Policy_1}, and 3) the deep learning policy provided in Section \ref{Sec_Learning}.
For Alg. \ref{A_RVI}, we set  
$N=10$, 
$I=2$,
 $\zeta=0.1$, 
 and $\varepsilon=0.001$.
For the deep learning policy,  we consider a fully-connected
deep neural network consisting of an input layer (${|\mathcal{Z}[t]|= 6+1=7}$ neurons), $2$ hidden layers consisting of $512$ and $256$ neurons with \textit{ReLU} activation function, and an output layer (${|\mathcal{A}|=9}$ neurons).
Moreover, the number of steps per episode is $600$, the discount factor is $0.99$, the mini-batch size is $64$,
the learning-rate is $0.0001$, and  the optimizer is \textit{RMSProp} \cite{RMSP}. 
\blue{The sources'  weights  are set to  $1$ for all sources.}
The system parameters, i.e.,  the arrival rates $\boldsymbol{\mu}=(\mu_1,\mu_2)$, the channel reliabilities $\bold{p}=(p_1,p_2)$, and the constraint budget 
$\Gamma_{\max}$ are specified  in the caption of each figure.
\\\indent
Next,  we provide  algorithm-specific analysis in Section \ref{Sec_NR_A} and performance comparison in Section \ref{Subsec_NR_B}.
\vspace{-1.5 em}
\subsection{Algorithm-specific Analysis}\label{Sec_NR_A}
\subsubsection{Algorithm \ref{A_RVI}} 
 Here, we verify Theorem \ref{Th_Detr} by visualizing the switching-type structure 
of  ${\lambda\text{-optimal}}$ policies and investigate the WS-AAoI  performance of the deterministic policy $\pi^*_{\lambda^+}$. 
\\\indent
 Fig. \ref{Fig_Str_Beta} shows the structure of a $\lambda$-optimal policy for the  decision at the relay $\beta$ with respect to the relative AoIs at the destination $y_1$ and $y_2$ for state  $\bold{s}=(1, 0, y_1, 2,1,y_2)$. The figure validates Theorem \ref{Th_Detr} and unveils that
 the relay schedules an available packet of the  source that has higher relative AoI at the destination; this is because the contribution of delivering such packet in the AoI reduction is higher than the other  who has a lower relative AoI at the destination. Having $\beta=0$ at  ($y_1=0$, $y_2=0$) is because   the most recent status update packets of the sources at the relay are also available at the destination; thus, resending them would  not reduce the AoI. 
 \\\indent
 Fig. \ref{Fig_Str_alpha} exemplifies the structure of the $\lambda$-optimal policy for the  decision at the transmitter $\alpha$ with respect to the relative AoIs at the relay $x_1$ and $x_2$ for state  $\bold{s}=(1, x_1,4, 1, x_2,4)$. Having $\alpha=0$ at ($x_1=0$, $x_2=1$) implies that transmission
does not occur
at every state due to the resource budget. Moreover, $\alpha=0$ at  ($x_1=0$, $x_2=0$) is because the most recent status update packets of the sources at the transmitter were already sent to the relay. 
The figure also shows that  for fixed $y_1$ and $y_2$, the transmitter will give a higher priority to schedule transmissions of Source 1 who has a lower arrival rate. That is, while the low packet arrival rate of Source 1 inevitably leads to infrequently receiving status updates by the destination, the optimal policy partly compensates for this by prioritizing to send fresh packets from Source 1, whenever possible.
\setlength{\textfloatsep}{0pt}
\begin{figure}[t]
\centering
\subfigure[
Decision on $\beta$ for different $y_1$ and $y_2$
]
{
\includegraphics[width=0.31\textwidth]{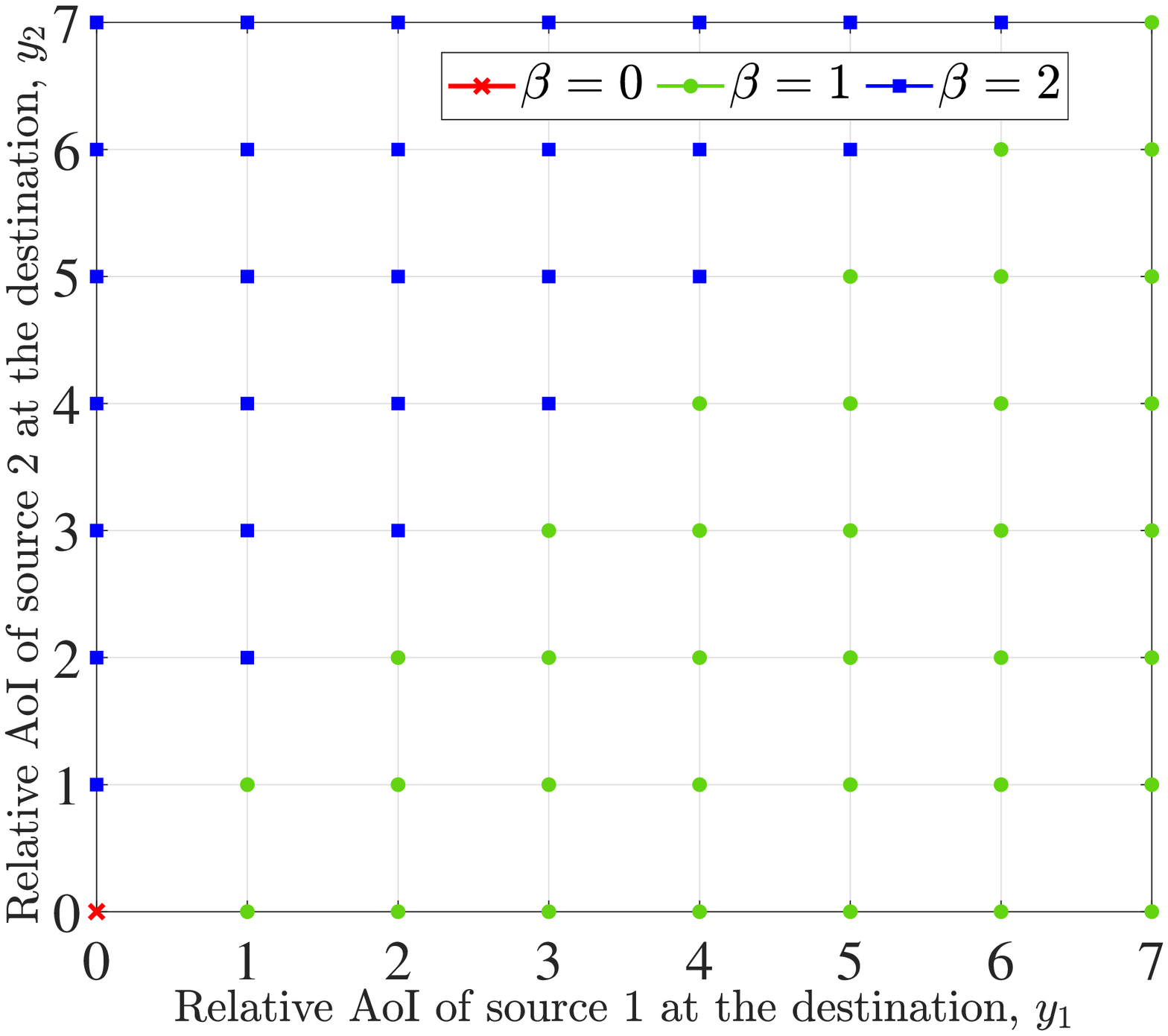}
\label{Fig_Str_Beta}
}
\subfigure[
Decision on $\alpha$ for different $x_1$ and $x_2$]{
\includegraphics[width=0.33\textwidth]{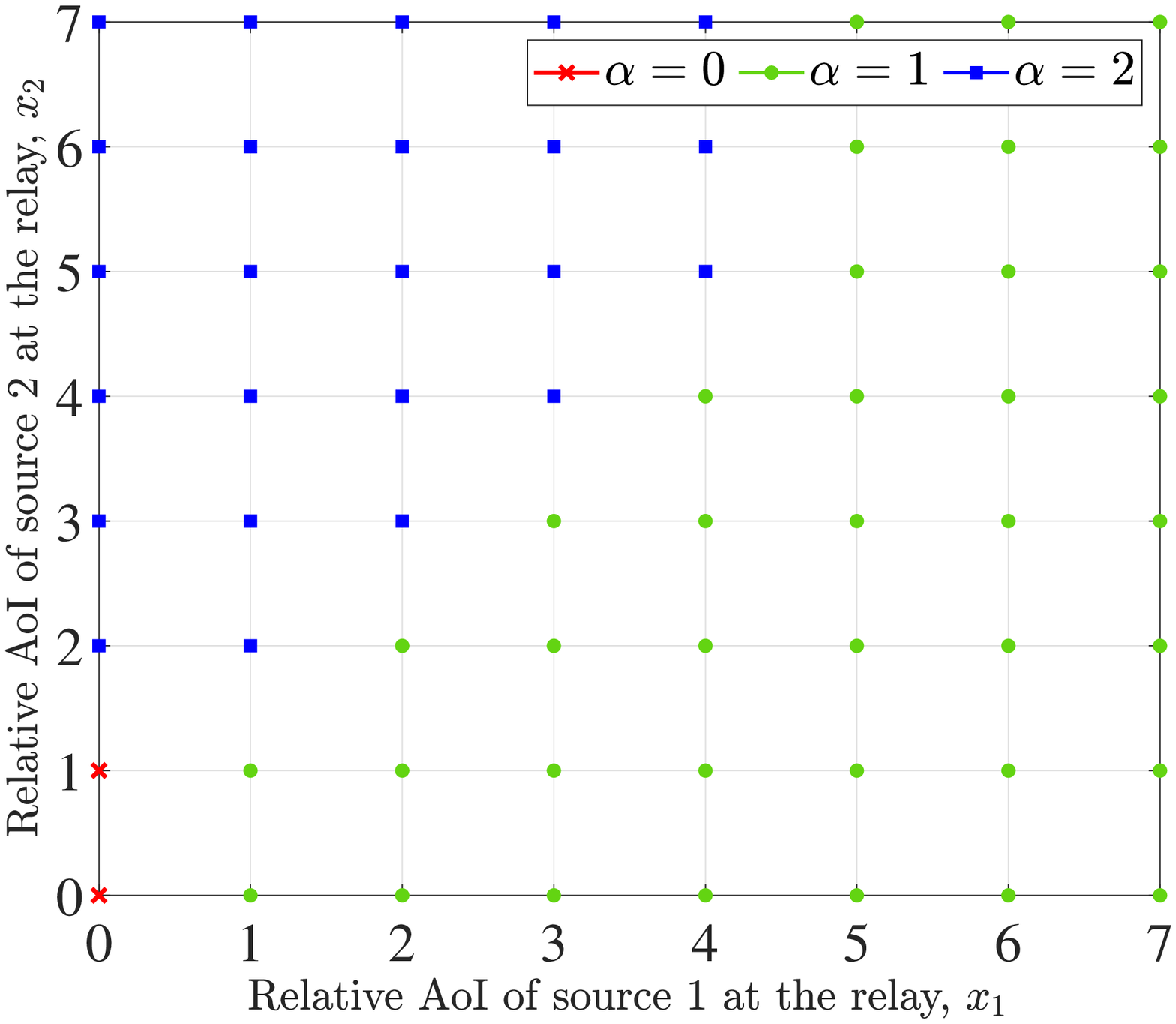}
\label{Fig_Str_alpha}}
\vspace{-1em}
\caption{An illustration of the switching-type structure, where $\lambda = 1.25,~\bold{p}=(0.8,0.7),~\text{and}~\boldsymbol{\mu}=(0.6, 0.9)$.}  
\label{Fig_ST_Strc}
\end{figure}
\\\indent Fig. \ref{Fig_RVI-BS_Gaamm} illustrates the WS-AAoI performance of the proposed policies obtained by Alg. \ref{A_RVI}
as a function of the constraint budget $\Gamma_{\max}$ 
obtained by  averaging over 100,000 time slots. 
The ``lower bound" is obtained by the infeasible policy $\pi^*_{\lambda^-}$. 
 First, Fig. \ref{Fig_RVI-BS_Gaamm} shows that the deterministic policy $\pi^*_{\lambda^+}$ achieves near-optimal performance and the lower bound is tight because  the difference between the feasible policy and the infeasible policy is small.   
In addition, we observe that the gap between the deterministic policy and the lower bound increases as $\Gamma_{\max}$ 
decreases. Thus, randomizing these two policies will produce the highest relative gain in this regime. 
\begin{figure}[t!]
    \centering
    \includegraphics[width=.45\textwidth]{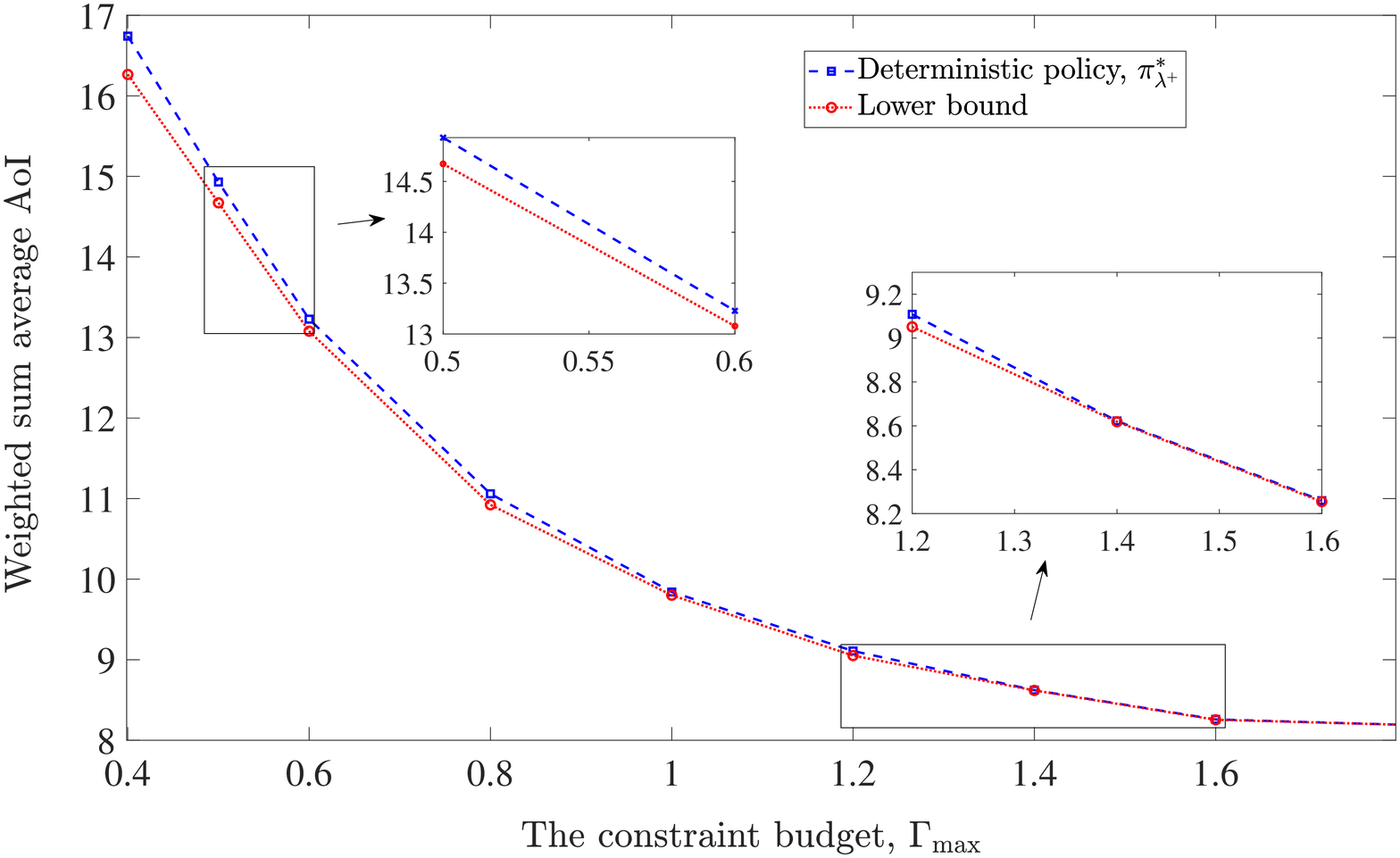}
    \vspace{- 0.5 em}
    \caption{The WS-AAoI versus $\Gamma_{\max}$ for the policies obtained by Alg. \ref{A_RVI}, where $\bold{p}=(0.7,0.8) $ and  $\boldsymbol{\mu}=(0.5,0.6)$.
    }
    \vspace{- 2 em}
    \label{Fig_RVI-BS_Gaamm}
\end{figure}
\subsubsection{DPP-SP} 
For DPP-SP, we investigate the impact of the trade-off  parameter $V$ on the WS-AAoI and  the average number of transmissions in the system 
in Fig. \ref{Fig_DPP_V}. 
Fig. \ref{Fig_SAOI_V} shows  the evolution of the WS-AAoI over time slots for different values of $V$. We observe that, for sufficiently small values of $V$, by increasing $V$, the WS-AAoI decreases. 
Fig. \ref{Fig_Cost_V} shows  
the evolution of the average number of transmissions over time slots.
 The figure validates Theorem \ref{Th_Stability} by showing that the time average constraint \eqref{Con_Sam1} is satisfied for all $V$. However,  the convergence speed decreases as $V$ increases. 
 These observations give us some practical guidelines in that we should set parameter $V$ large (but not excessively high) to obtain a low value of the WS-AAoI, because increasing $V$ beyond a certain value does not bring significant improvements. 
\begin{figure}[t]
\centering
\subfigure[Evolution of the WS-AAoI versus time slots]
{
\includegraphics[width=0.38\textwidth]{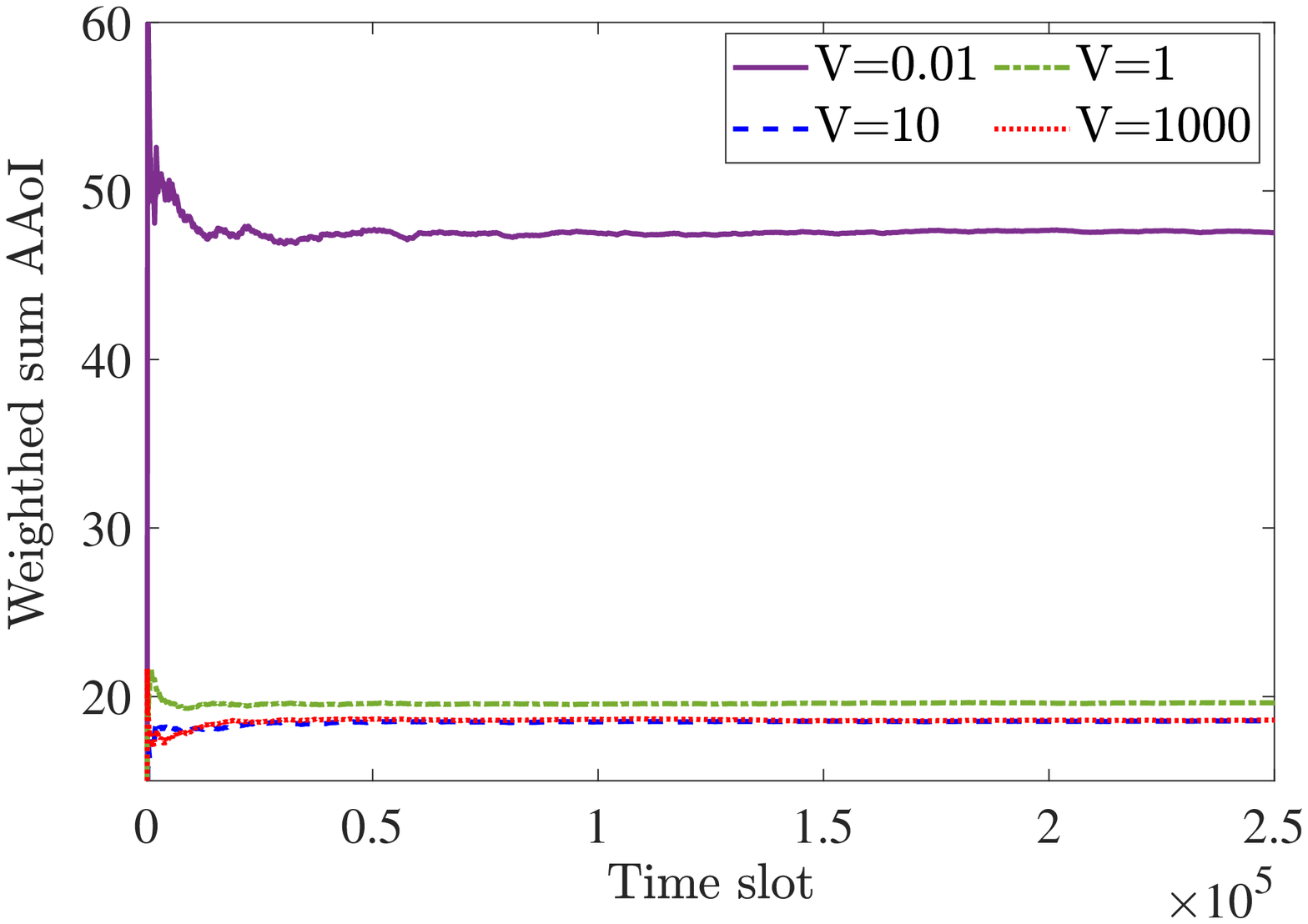}
\label{Fig_SAOI_V}
}
\subfigure[Evolution of the  average number of transmissions 
versus time slots]{
\includegraphics[width=0.36\textwidth]{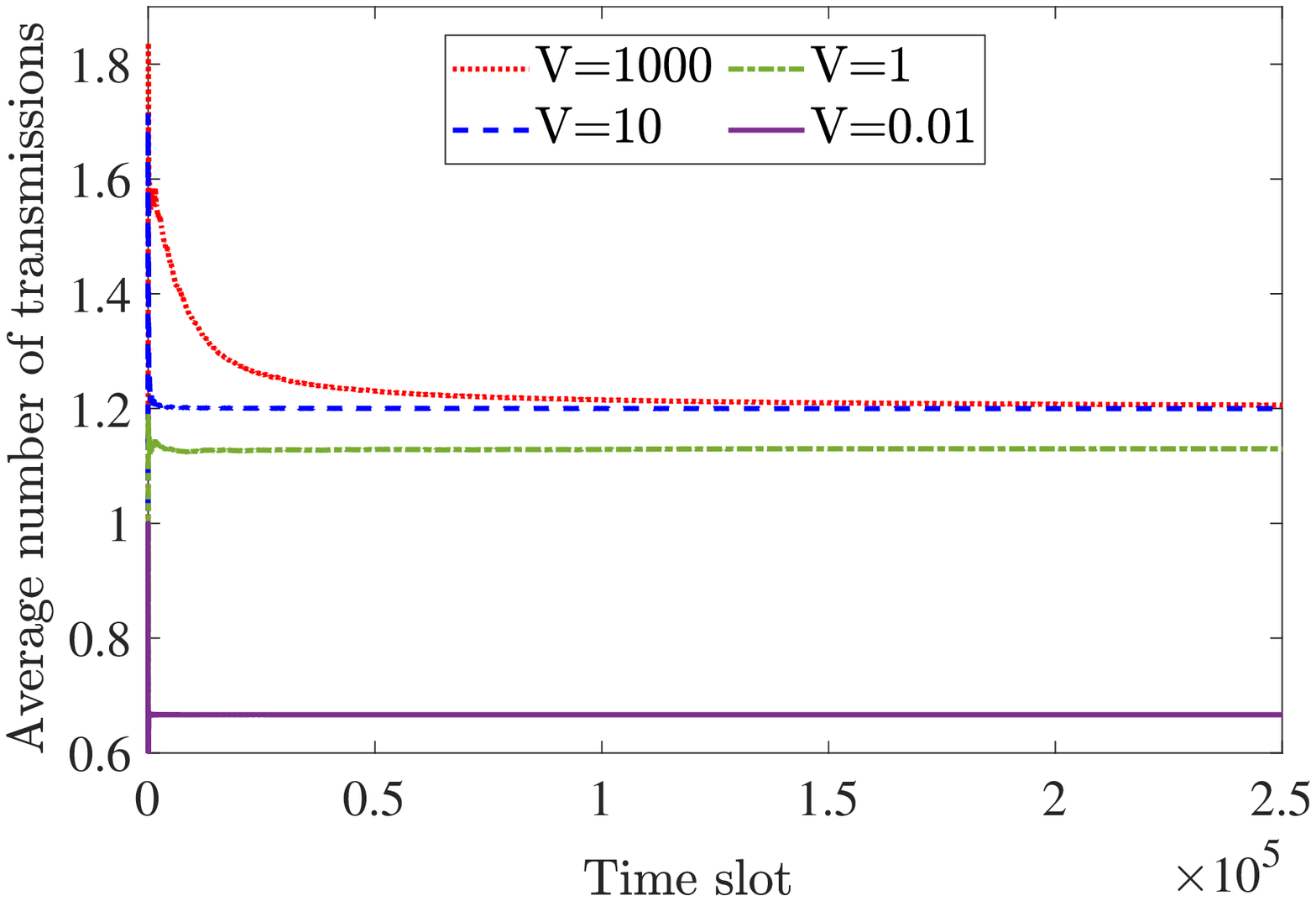}
\label{Fig_Cost_V}}
\vspace{-1em}
\caption{Impact of parameter $V$ on  DPP-SP, where $\Gamma_{\max}=1.2,~\bold{p}=(0.3,0.4),~\text{and}~\boldsymbol{\mu}=(0.5, 0.7)$.}
\label{Fig_DPP_V}
\end{figure}
\subsubsection{Deep Learning Policy} 
For the deep learning policy, we show  the evolution of
the \textit{episodic reward} over episodes in Fig. \ref{Fig_D3QN_reward}, the evolution of the average number of transmissions over episodes in Fig. \ref{Fig_D3QN_Cost}, and the evolution of the WS-AAoI over episodes  in Fig. \ref{Fig_D3QN_S-AAoI} for different values of $\Gamma_{\max}$. The episodic reward is defined by the sum of rewards obtained at each episode. 
 Fig. \ref{Fig_D3QN_Cost} validates that the proposed  deep learning policy satisfies the time average constraint \eqref{Con_Sam1} for all the  constraint budgets.
However, the  convergence speed is highly affected by $\Gamma_{\max}$, i.e., as $\Gamma_{\max}$  increases, the policy  converges quickly. 
The same convergence behavior is seen for the episodic reward function in Fig. \ref{Fig_D3QN_reward} and the WS-AAoI   in Fig. \ref{Fig_D3QN_S-AAoI}. 
\begin{figure}[t]
\centering
\subfigure[Evolution of the episodic reward versus episodes]
{
\includegraphics[width=0.31\textwidth]{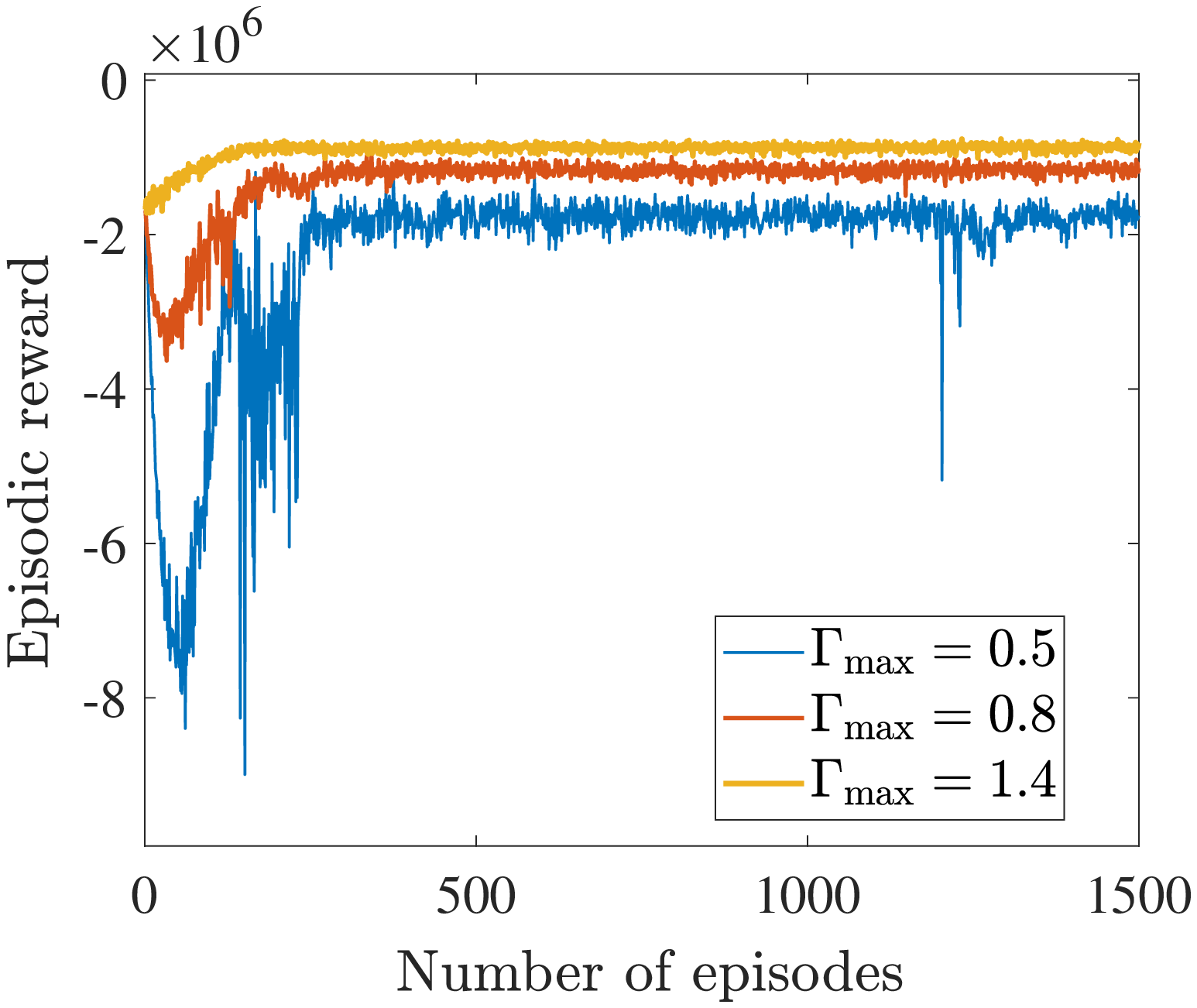}
\label{Fig_D3QN_reward}
}
\subfigure[Evolution of the  average number of transmissions  versus episodes]{
\includegraphics[width=0.305\textwidth]{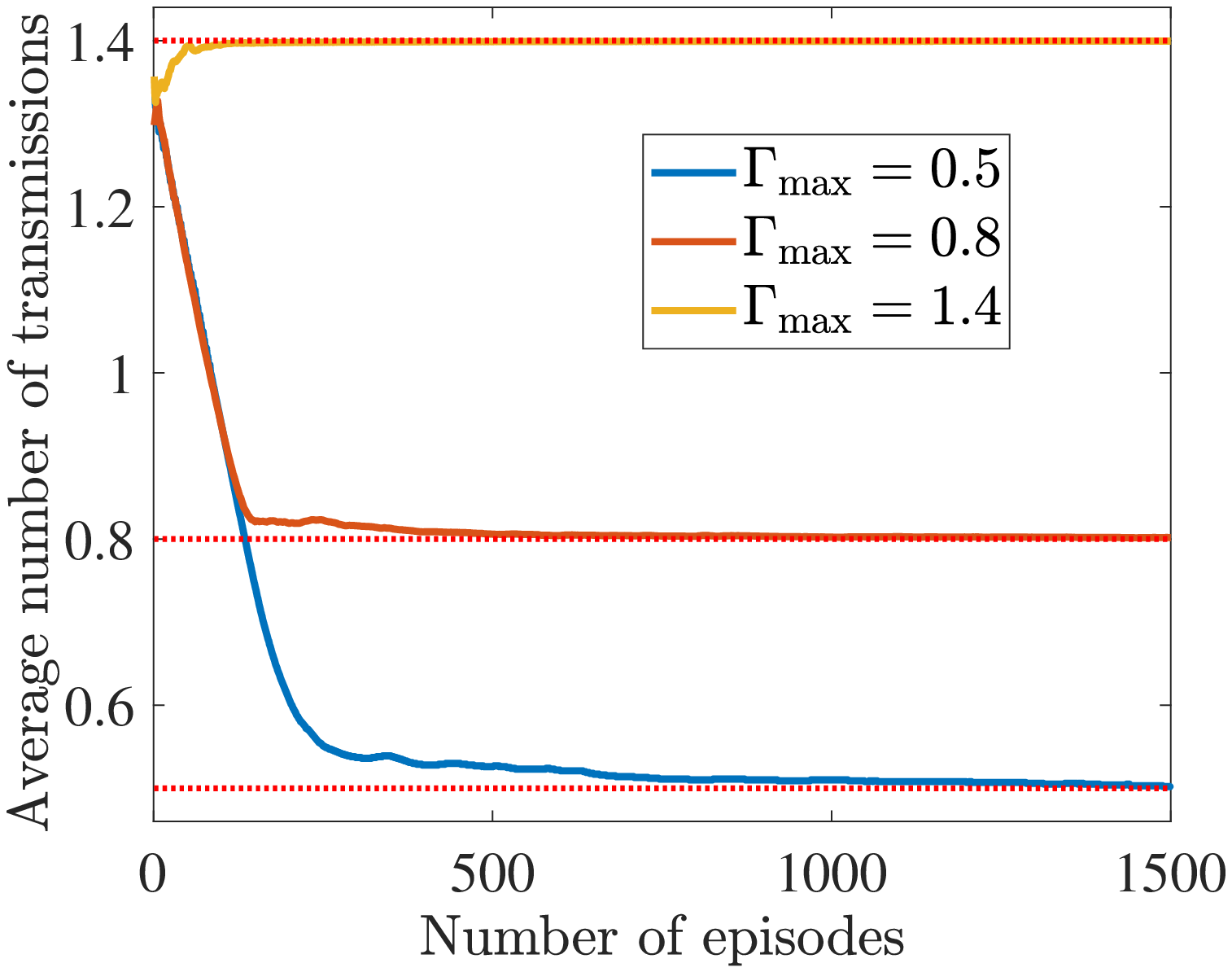}
\label{Fig_D3QN_Cost}
}
\subfigure[Evolution of the WS-AAoI versus episodes]{
\includegraphics[width=0.30\textwidth]{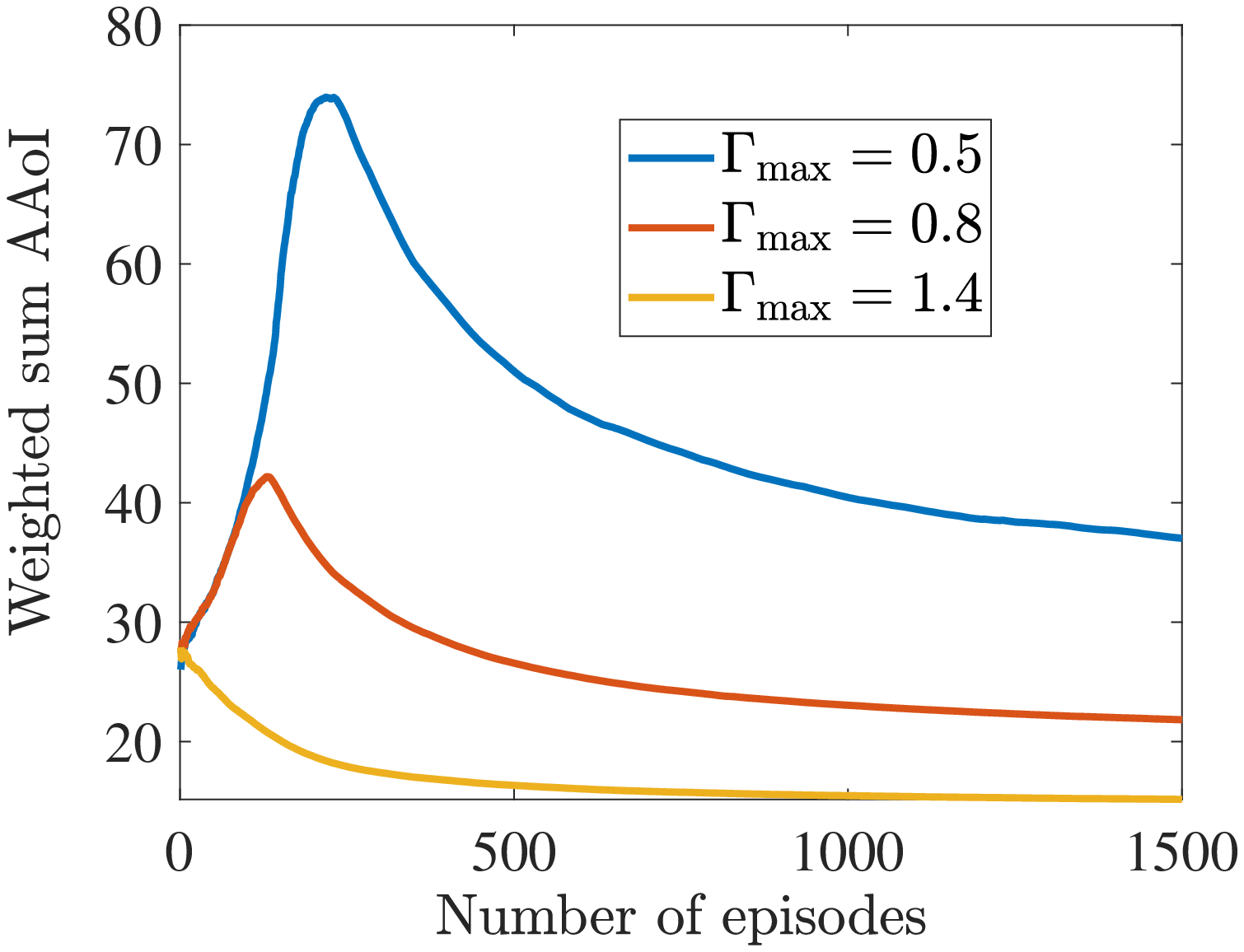}
\label{Fig_D3QN_S-AAoI}
}
\vspace{- 1 em}
\caption{Results of the deep learning policy for different $\Gamma_{\max}$, where  $V=100,~\bold{p}=(0.4,0.5),~\text{and}~\boldsymbol{\mu}=(0.6, 0.8)$.}
\label{Fig_D3QN_all}
\end{figure}
\vspace{-1.5em}
\subsection{Performance Comparisons} \label{Subsec_NR_B}
In this subsection, we provide a performance comparison of the proposed policies.  The results are averaged over 100,000 time slots and the parameter $V$ is set to $100$.  
 For comparison, we also consider a greedy ``baseline policy", which determines the transmission decision variables at each slot $t$ according to the following rule:   
\textit{
    If $\bar{D}_t\le \Gamma_{\max}$, then ${\alpha[t]= \argmax_i x_i[t]}$ and ${\beta[t] = \argmax_i y_i[t]}$; otherwise, $\alpha[t]= 0$ and $\beta[t]= 0$,}
     where $\bar{D}_t$ denotes the average number of transmissions until slot $t$. This policy satisfies the time average constraint \eqref{Con_Sam1}.
     It is remarkable that the baseline policy and DPP-SP (given by \eqref{Eq_MW_Policy_1}) have similar computational complexity.
\subsubsection{Effect of the Constraint Budget}  Fig. \ref{Fig_Com_gamma} depicts the WS-AAoI performance of the proposed policies and the baseline policy as a function of  the constraint budget $\Gamma_{\max}$. 
First, Fig. \ref{Fig_Com_gamma} reveals that the low-complexity DPP-SP has  near-optimal performance because it 
nearly coincides with  the (near-optimal) RVIA-based deterministic policy $\pi_{\lambda^+}^*$ obtained by Alg. \ref{A_RVI}.
The figure also shows that the  deep learning policy  obtains near-optimal performance when  the constraint budget becomes sufficiently large, e.g., ${\Gamma_{\max}\ge 0.8}$. 
Moreover, the figure shows that 
the WS-AAoI performance gap between the baseline  policy and the  proposed policies 
is extremely large when the constraint budget is small;
this is because in such cases, performing good actions in each slot becomes more critical due to having a high limitation on the average number of transmissions.
The figure shows that the proposed policies achieve
up to almost $91\%$ improvement in the WS-AAoI performance
   compared to the baseline policy. Finally, we can observe that, as the constraint budget increases, the WS-AAoI values decrease; however,  from a certain point onward, increasing the constraint budget does not considerably decrease the WS-AAoI.
\begin{figure}[t] 
    \centering
    \includegraphics[width=.5\textwidth]{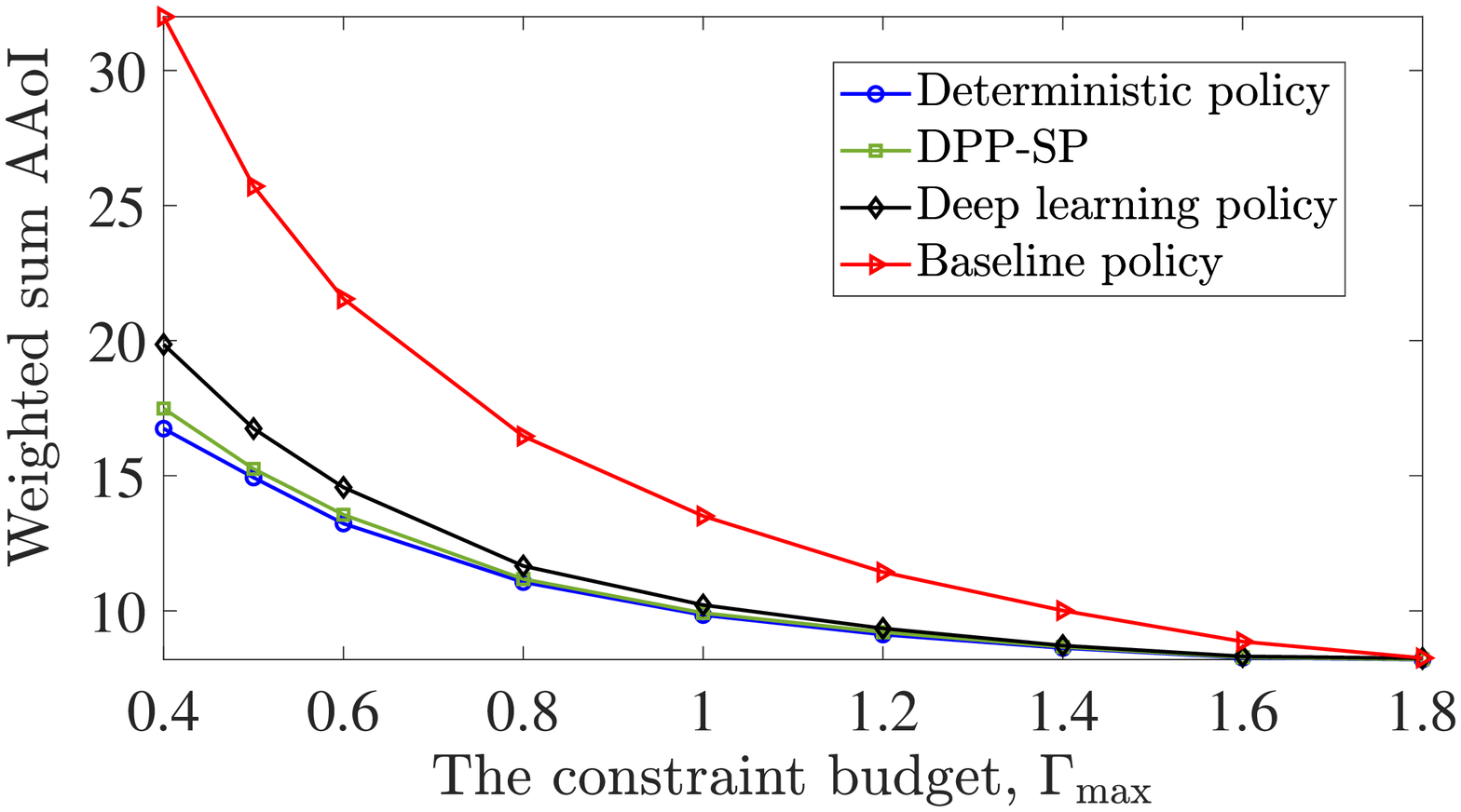}
    \vspace{-1.5em}
    \caption{The WS-AAoI versus the constraint budget for the different  policies, where  $\bold{p}=(0.7,0.8)$ and $\boldsymbol{\mu}=(0.5,0.6)$. 
    }
    \vspace{- 1 em}
    \label{Fig_Com_gamma}
\end{figure}
\subsubsection{Effect of the Arrival Rates}
In Fig. \ref{Fig_SAAoI_AR},  we examine the impact of the arrival rates $\mu_1$ and $\mu_2$ on the WS-AAoI performance of the different policies.
\blue{The figure shows  that the WS-AAoI increases as the arrival rates decrease. This is because when the arrival rates decrease, the rate of fresh update delivery at the destination decreases.}
The figure also reveals that, as the arrival rates increase,  the reduction of the WS-AAoI by the proposed policies in comparison to the baseline policy  becomes increasingly more prominent. 
\blue{The reason for this behavior is that by the increase of the arrival rates there are more new fresh packets which  can potentially reduce the AoI if they are delivered timely/optimally to the destination. The greedy baseline policy, however, cannot  deliver them timely.  }
Moreover, it is observable that when the arrival rates are sufficiently large, increasing them further does not considerably reduce the WS-AAoI.
\blue{This observation is due to the fact that in our system, only one packet can be transmitted in each slot, and for large values of the arrival rates, the probability of having at least one fresh packet does not change considerably by changing the arrival rates.}
\subsubsection{Effect of the Successful Transmission Probabilities}
In Fig. \ref{Fig_SAAoI_CR}, we examine the impact of the
successful transmission probabilities  $p_1$ and $p_2$ 
on 
the WS-AAoI performance of the different policies. 
 First, the figure shows that the WS-AAoI performance gap between the proposed policies 
and  the baseline policy is significant, especially when the successful transmission probabilities are small.
\blue{The reason is that when the successful transmission probabilities are small, finding optimal transmission times become more critical, as there are resource limitations.  }
Moreover, the figure shows that the WS-AAoI considerably decreases as the successful transmission probabilities increase; this is expected, \blue{because
 the probabilities of successfully receiving the transmitted status update packets through the unreliable links increase, and consequently,  the destination receives updates more frequently.  }
\begin{figure}[t]
\centering
\subfigure[The WS-AAoI versus the arrival rates for $\bold{p}=(0.6,0.7)$]
{
\includegraphics[width=0.38\textwidth]{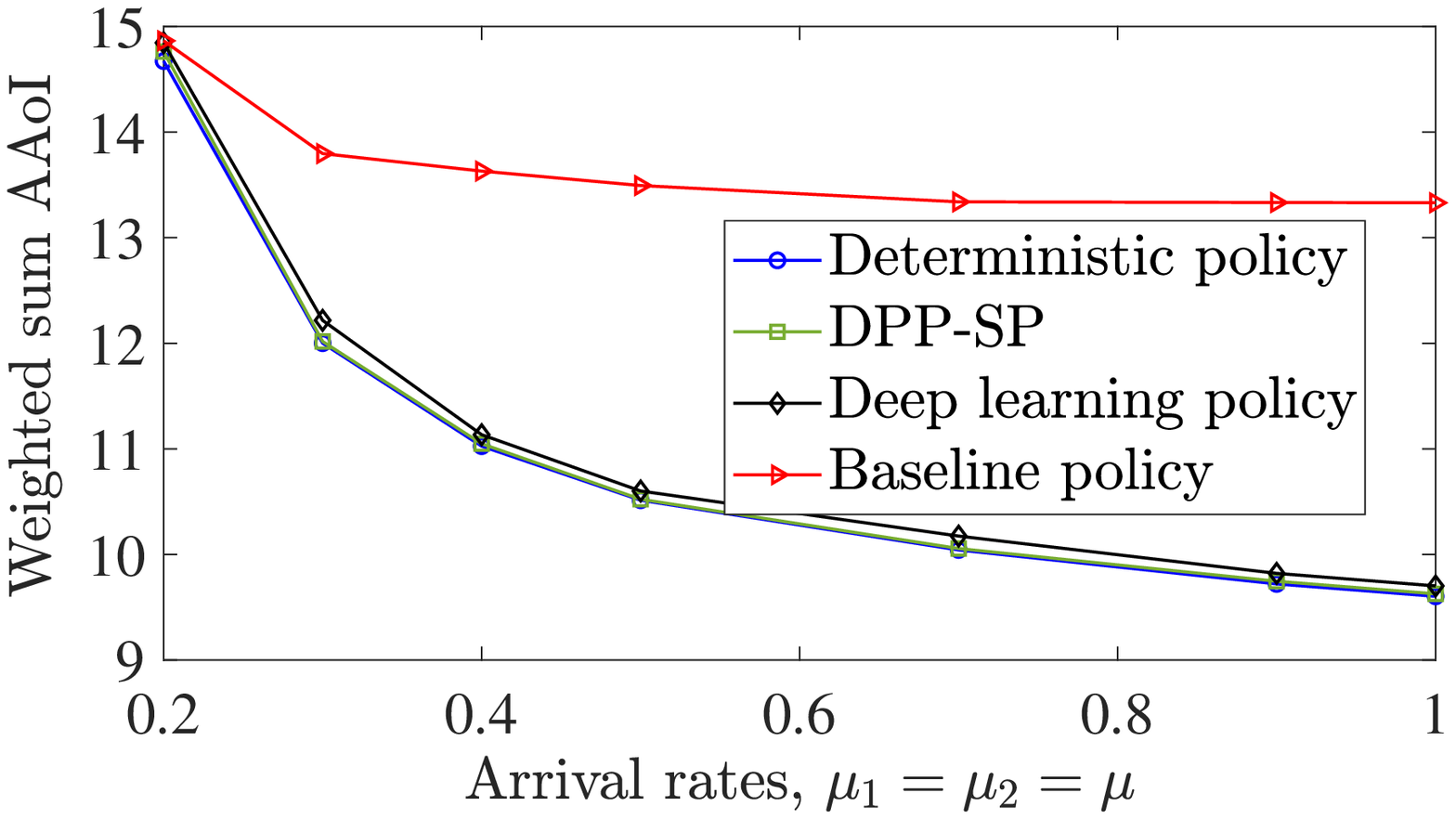}
\label{Fig_SAAoI_AR}
}
\subfigure[The WS-AAoI versus the successful transmission probabilities for $\boldsymbol{\mu}=(0.6, 0.7)$]{
\includegraphics[width=0.38\textwidth]{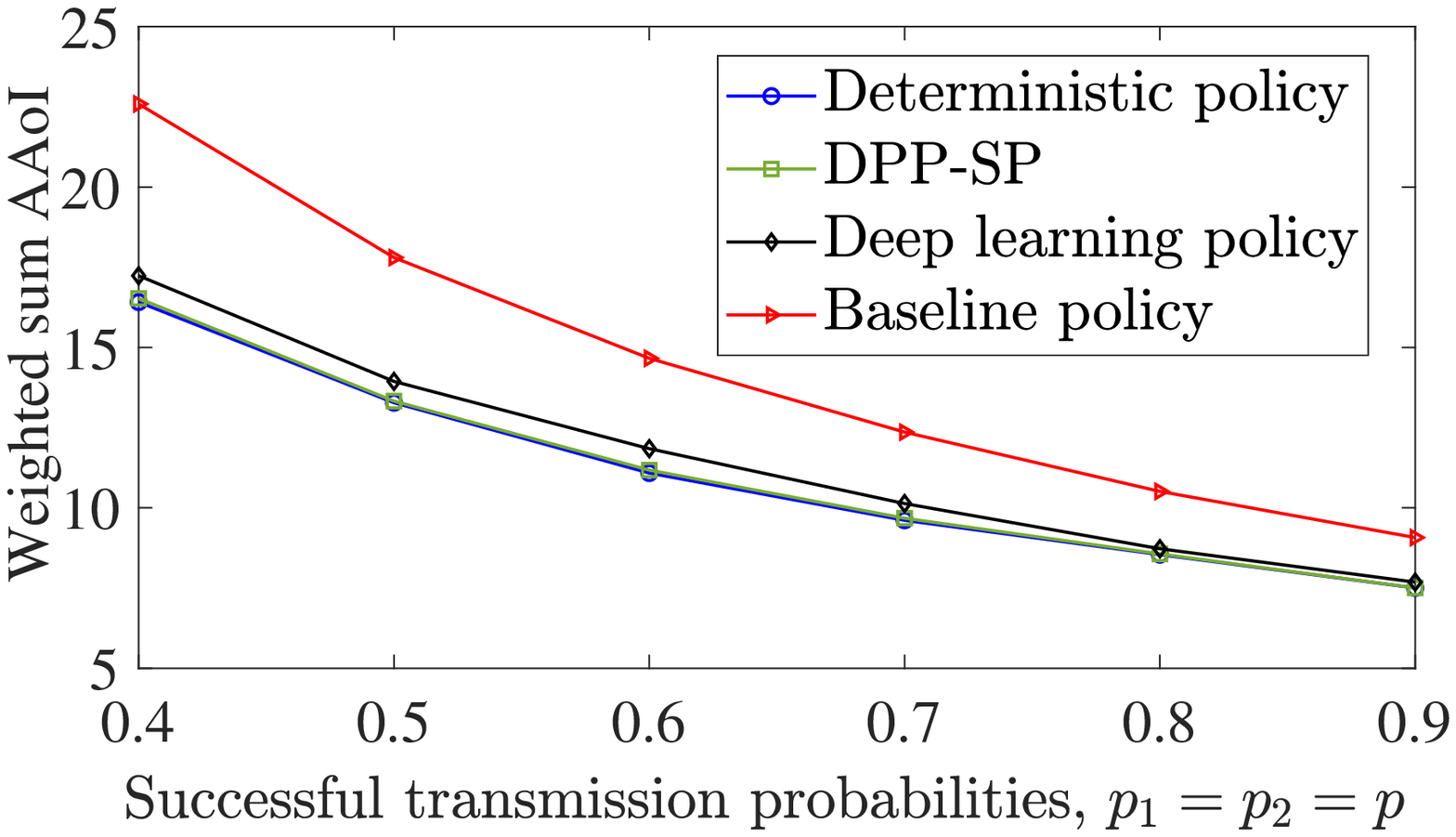}
\label{Fig_SAAoI_CR}}
\vspace{-1em}
\caption{The WS-AAoI performance of the proposed and baseline  policies, where $\Gamma_{\max}=1.2$.
}
\label{Fig_AR_CR}
\end{figure}
\subsubsection{Effect of Number of Sources}
In Fig. \ref{Fig_Source}, we show the effect of the  number of sources on the WS-AAoI for different values of the constraint budget $\Gamma_{\max}$ without bounding the AoI.
\blue{Here, we utilize the DPP-SP, the deep learning policy,  and the greedy baseline policy; notably, as explained in Section \ref{Sec_ComAna}, Alg. \ref{A_RVI} is not scalable to a  multi-source setup (with a high number of sources)}. The figure shows that WS-AAoI increases by increasing $I$. \blue{This is because, for a fixed $\Gamma_{\max}$,  when $I$ increases, the opportunity of having  transmissions for each source decreases; thus, the WS-AAoI increases.}
\begin{figure}[t] 
    \centering
    \includegraphics[width=.38\textwidth]{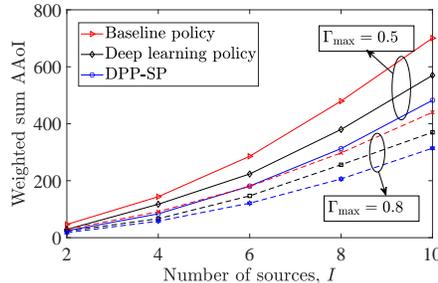}
    \vspace{-1em}
    \caption{\blue{The WS-AAoI versus the number of sources, where  $\bold{p}=(0.4,0.5)$ and $\mu_i = 0.6$ for all sources.} 
    }
    \label{Fig_Source}
\end{figure}
\subsubsection{\blue{Effect of the Source's Weight}}
\blue{{Fig. \ref{Fig_Weight} illustrates the impact of the weight $w_i$ on the AAoI of sources for DPP-SP in a two-source setup. As can be seen, by increasing the weight of a source, its AAoI decreases, as expected. The reason is that by increasing the weight of a source, we put more emphasis on the AoI of the source, and thus, the policy tries to keep its AoI lower.}
}
\begin{figure}[t] 
    \centering
    \includegraphics[width=.32\textwidth]{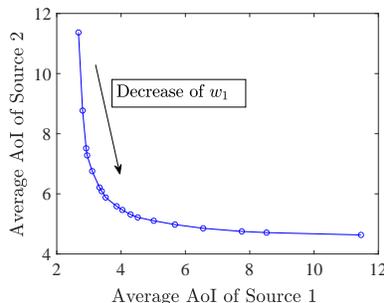}
    \vspace{-1.5em}
    \caption{
    \blue{Impact of weight $w_i$ on the AAoI of sources in a two-source setup, where $w_2 = 1- w_1$ and ${ w_1 = [0.05:0.05:0.95] }$. 
    Moreover, $\bold{p}=(0.7,0.7)$ and $\boldsymbol{\mu}=(0.6, 0.6)$.
    }
    }
    \label{Fig_Weight}
\end{figure}

\vspace{-2.5em}
\section{Conclusion}\label{Sec_Conclusion}
We studied the WS-AAoI minimization problem in a multi-source relaying system with stochastic arrivals and unreliable channels subject to transmission capacity  and  the average number of  transmissions constraints.
We formulated a stochastic optimization problem and solved it with three different algorithms. 
Specifically, we proposed the CMDP approach in which 
we first conducted analysis to show that an optimal policy of the MDP problem has a switching-type structure and subsequently, utilized this structure to devise a
structure-aware RVIA that gives a near-optimal  deterministic policy and a tight lower bound; the convergence of the algorithm was also proven.
We devised a dynamic near-optimal low-complexity DPP-SP, representing an efficient online scheduler for systems with
large numbers of sources. 
Moreover, we  devised a deep learning policy combining  the Lyapunov optimization theory and D3QN.
\\\indent
We numerically investigated the effect of system parameters 
on the WS-AAoI and showed the effectiveness of our proposed policies compared to the baseline policy; the results showed up to $91\%$ improvement in the WS-AAoI performance.
Accordingly, an age-optimal scheduler design is crucial for resource-constrained relaying status update systems, where greedy-based scheduling  is inefficient. 
Moreover, the results showed that the proposed deep learning policy satisfies the time average constraint and achieves performance close to  the other proposed near-optimal policies  in many settings. 
\vspace{-1 em}
\appendix
\vspace{- 1.5 em}
\subsection{Proof of Theorem \ref{The_Unichain}}\label{Proof_Unichain}
By \cite[Exercise 4.3]{Gallager_SP_Theory_App}, it is sufficient to show that the Markov chain, described by the transition probability matrix with elements $\mathcal{P}_{\bold{s}\bold{s}'}(\bold{a})$, corresponding to  every deterministic policy has a state which is accessible from any other state.  We show this 
by dividing the sources into two different groups $\mathcal{I}_1$ and $\mathcal{I}_2$ based on the values of the arrival rates $\mu_i$, i.e., sources with $\mu_i=1$ belongs to $\mathcal{I}_1$ and sources with $\mu_i\in(0,1)$ belongs to $\mathcal{I}_2$. 
Let us express each  state $\bold{s}\in\mathcal{S}$ by  $\bold{s}=\{\bold{s}_i\}_{i\in\mathcal{I}_1\cup\mathcal{I}_2}$, where recall that $ { \bold{s}_i = (\theta_i, x_i, y_i) }$.
Then, in the Markov chain induced by every deterministic policy, state $\bold{s}^{\mathrm{acc}}=\{\bold{s}^{\mathrm{acc}}_i\}_{i\in\mathcal{I}_1\cup\mathcal{I}_2},$ where $\bold{s}^{\mathrm{acc}}_i=(0,N,0), \forall\,i\in\mathcal{I}_1$ and $\bold{s}^{\mathrm{acc}}_i=(N,0,0), \forall\,i\in\mathcal{I}_2$, 
is accessible from any other state. 
This  is due to the fact that regardless of actions taken: (1) there is always new arrivals for sources belong to $\mathcal{I}_1$, (2) the probability of having no arrivals for all sources belong to $\mathcal{I}_2$,
 for at least $N$ consecutive  slots is $\prod_{i} (1-\mu_i)^N,~i\in\mathcal{I}_2$, which is positive,  and (3)
 the probability of having
 unsuccessful receptions  in both the relay and the destination 
  for at least $N$ consecutive  slots is $(1-p_1)^N(1-p_2)^N$, which is positive. 
Thus, according to the evolution of the AoIs, starting from any state at any slot $t$ leads to state $\bold{s}^{\mathrm{acc}}$ with a positive  probability, which completes the proof.
\vspace{-1 em}
\subsection{Proof of Theorem \ref{Th_Detr}}\label{Proof_App_Threshod}
To show the switching-type structure w.r.t. $y$ for a  $\lambda$-optimal policy,  we use Theorem \ref{The_Bellman}.  
First, by turning the optimality equation \eqref{Eq_Bellman} into the iterative procedure \eqref{Eq_RVI}, for each state  $\bold{s}\in \mathcal{S}$, we can iteratively obtain ${h(\bold{s})\triangleq V(\bold{s}) - V(\bold{s}_{\mathrm{ref}})}$ 
and consequently $\pi^*_{\lambda}(\bold{s})$ (see \eqref{Eq_Optimal_policy}).
We then use \eqref{Eq_RVI} and show
  a monotonic property  of  the function $V(\bold{s})$ in the following lemma, which will be used in the next steps of the proof. 
  \vspace{-0.5 em }
\begin{Lem}\label{Lemm_Non-dec}
The function $V(\bold{s})$ is a non-decreasing function with respect to every $s_j$, where $s_j,~j=1,\dots,3I,$ is the $j$-th element of state vector $\bold{s}=(\theta_1,x_1,y_1,\dots,\theta_I,x_I,y_I)$.
\end{Lem}
\begin{proof}
The proof is based on the induction hypothesis.
 First,  the sequence ${\{V_{n}(\bold{s})\}_{n=1,2,\dots}}$, updated by \eqref{Eq_RVI},  converges to $ V(\bold{s})$ for any initialization (see Theorem \ref{The_RVIA_Converg}). Also, Lemma \ref{Lemm_Non-dec} holds  for $V_{0}(\bold{s})$. Now, we assume that $V_{n}(\bold{s})$ is non-decreasing in $s_j$. The
 immediate cost of the MDP 
 ${L(\bold{s},\bold{a};\lambda)={\sum}_{i=1}^{I}  w_i(\theta_i+x_i+y_i)+\lambda\left( D(\bold{a}[t]) - \Gamma_{\max}\right)}$
 is a non-decreasing function in ${s_j,~j=1,\dots,3I}$. In addition, $
 {\textstyle\sum_{\bold{s}'\in\mathcal{S}}}\mathcal{P}_{\bold{s}\bold{s}'}(\bold{a})V_n(\bold{s}')$ is a non-decreasing function in  $s_j$ via the induction hypothesis, and the minimum operator in \eqref{Eq_RVI} preserves the non-decreasing property.
 Thus, we  conclude that $V_{n+1}(\bold{s})$ is  non-decreasing in every $s_j$, which completes the proof.
\end{proof}
 Now, we need to show that if for an arbitrary state $\bold{s}$ a $\lambda$-optimal policy takes action $\beta=i$, then  for all states $\bold{s}+k\bold{e}_{3i}$ the policy takes also action $\beta=i$, where $k$ is a positive integer.
 Let us define a function $\mathcal{V}(\bold{s},\bold{a};\lambda) \triangleq L(\bold{s},\bold{a};\lambda)+\Bbb{E}\left\{V(\bold{s}')~\big|~\bold{s},\bold{a}\right\} - V(\bold{s}_{\mathrm{ref}})$, where $\Bbb{E}\left\{V(\bold{s}')~\big|~\bold{s},\bold{a}\right\} \triangleq {\textstyle\sum_{\bold{s}'\in\mathcal{S}}}\mathcal{P}_{\bold{s}\bold{s}'}(\bold{a})V(\bold{s}')$. Without loss of generality, suppose that $I=2$ (for notation simplicity) and  an optimal policy  takes action  $\bold{a}=(2,1)$ at state $\bold{s}$ which implies the following: 
 \begin{equation}\label{Eq_Assum_NOndec}
 \begin{array}{ll}
    \mathcal{V}\big(\bold{s},(2,1);\lambda\big)\le  \mathcal{V}\big(\bold{s},(2,2);\lambda\big) \Rightarrow \Bbb{E}\left\{V(\bold{s}')~\big|~\bold{s},(2,1)\right\}\le \Bbb{E}\left\{V(\bold{s}')~\big|~\bold{s},(2,2)\right\}.  
    \end{array}
 \end{equation}
 \indent
 To show the switching-type structure w.r.t. $y$, we must show that 
 $ {\Bbb{E}\left\{V(\bold{s}')~\big|~\bar{\bold{s}},(2,1)\right\}\le \Bbb{E}\left\{V(\bold{s}')~\big|~\bar{\bold{s}},(2,2)\right\}}$, where $\bar{\bold{s}}=(\theta_1,x_1,y_1+k,\theta_2,x_2,y_2)$. 
 Let us express $V(\bold{s})$ as $V({\bold{s}_1},{\bold{s}_2})$ (with slight abuse of notation), where ${\bold{s}_1}=(\theta_1,x_1,y_1)$ and ${\bold{s}_2}=(\theta_2,x_2,y_2)$.
 We calculate the expectations as follows:
 \begin{equation}\label{Eq_Str_Beta1}
 \allowdisplaybreaks
 \nonumber
 \begin{array}{ll}
    \Bbb{E}\left\{V(\bold{s}')~\big|~\bar{\bold{s}},(2,1)\right\}=
     \mu_1\mu_2p_1 p_2V\left(\bold{s}_{1,1},\bold{s}_{2,1}\right)
     +
     \mu_1\mu_2(1-p_1) p_2V\left(\bold{s}_{1,1},\bold{s}_{2,2}\right)
          \\
     +
     \mu_1(1-\mu_2)p_1 p_2V\left(\bold{s}_{1,1},\bold{s}_{2,3}\right)
     + \mu_1(1-\mu_2)(1-p_1) p_2V\left(\bold{s}_{1,1},\bold{s}_{2,4}\right)
     + 
     \mu_1\mu_2p_1 (1-p_2) \\ V\left(\bold{s}_{1,2},\bold{s}_{2,1}\right)
     +  
     \mu_1\mu_2(1-p_1) (1-p_2)V\left(\bold{s}_{1,2},\bold{s}_{2,2}\right)
     + \mu_1(1-\mu_2)p_1 (1-p_2)V\left(\bold{s}_{1,2},\bold{s}_{2,3}\right)
    \\
     + 
     \mu_1(1-\mu_2)(1-p_1) (1-p_2)V\left(\bold{s}_{1,2},\bold{s}_{2,4}\right)
     + ((1-\mu_1)\mu_2p_1 p_2)V\left(\bold{s}_{1,3},\bold{s}_{2,1}\right)
     \\
     +(1-\mu_1)\mu_2(1-p_1) p_2V\left(\bold{s}_{1,3},\bold{s}_{2,2}\right)
     + (1-\mu_1)(1-\mu_2) p_1 p_2V\left(\bold{s}_{1,3},\bold{s}_{2,3}\right)
     \\
     +  (1-\mu_1)(1-\mu_2) (1-p_1) p_2V\left(\bold{s}_{1,3},\bold{s}_{2,4}\right)
     +
(1-\mu_1)\mu_2 p_1 (1-p_2)V\left(\bold{s}_{1,4},\bold{s}_{2,1}\right)
\\
+(1-\mu_1)\mu_2 (1-p_1) (1-p_2)V\left(\bold{s}_{1,4},\bold{s}_{2,2}\right)
+(1-\mu_1)(1-\mu_2) p_1 (1-p_2)V\left(\bold{s}_{1,4},\bold{s}_{2,3}\right)
\\
+ (1-\mu_1)(1-\mu_2) (1-p_1) (1-p_2)V\left(\bold{s}_{1,4},\bold{s}_{2,4}\right),
     \end{array}
 \end{equation}
 where
$\bold{s}_{1,1}=(0,\tilde{x}_1+\tilde{\theta}_1,0),~\bold{s}_{1,2}=(0,\tilde{x}_1+\tilde{\theta}_1,y'_1),~\bold{s}_{1,3}=(\tilde{\theta}_1,\tilde{x}_1,0),~\bold{s}_{1,4}=(\tilde{\theta}_1,\tilde{x}_1,y'_1),~
\bold{s}_{2,1}=(0,\tilde{\theta}_2,\tilde{y}_2+\tilde{x}_2),~\bold{s}_{2,2}=(0,\tilde{x}_2+\tilde{\theta}_2,\tilde{y}_2),~\bold{s}_{2,3}=(\tilde{\theta}_2,0,\tilde{y}_2+\tilde{x}_2),~\bold{s}_{2,4}=(\tilde{\theta}_2,\tilde{x}_2,\tilde{y}_2),$
where $y'_1 = \min\big(y_1+k+x_1+\theta_1+1, N\big)-\min\big(x_1+\theta_1+1, N\big)$; recall that $\tilde{\theta}_i,\,\tilde{x}_i,$ and $\tilde{y}_i$ were defined (see Section \ref{Sec_CMDP_1}) as follows: $ \tilde{\theta}_i \triangleq \min\big(\theta_i+1, N\big)$,
 $\tilde{x}_i \triangleq  \min\big(x_i+\theta_i+1, N\big)-\min\big(\theta_i+1, N\big)$, and  
 $\tilde{y}_i \triangleq \min\big(y_i+x_i+\theta_i+1, N\big)-\min\big(x_i+\theta_i+1, N\big)$.
We calculate $\Bbb{E}\left\{V(\bold{s}')~\big|~\bar{\bold{s}},(2,2)\right\}$ 
by
\begin{equation}\label{Eq_Str_Beta2}
\nonumber
    \begin{array}{ll}
   \Bbb{E}\left\{V(\bold{s}')~\big|~\bar{\bold{s}},(2,2)\right\}=\mu_1\mu_2p_1p_2V\left(\bold{s}_{1,2},\bold{s}'_{2,1}\right)
   + \mu_1\mu_2(1-p_1)p_2V\left(\bold{s}_{1,2},\bold{s}'_{2,2}\right)
   + 
    \\
   \mu_1(1-\mu_2)p_1p_2V\left(\bold{s}_{1,2},\bold{s}'_{2,3}\right)
   + \mu_1(1-\mu_2)(1-p_1)p_2V\left(\bold{s}_{1,2},\bold{s}'_{2,4}\right)
   +\mu_1\mu_2p_1(1-p_2) \\
   V\left(\bold{s}_{1,2},\bold{s}_{2,1}\right)
     + 
   \mu_1\mu_2(1-p_1)(1-p_2)V\left(\bold{s}_{1,2},\bold{s}_{2,2}\right)
   + \mu_1(1-\mu_2)p_1(1-p_2)V\left(\bold{s}_{1,2},\bold{s}_{2,3}\right)
   \\
   + \mu_1(1-\mu_2)(1-p_1)(1-p_2)V\left(\bold{s}_{1,2},\bold{s}_{2,4}\right)
   +
   (1-\mu_1)\mu_2p_1p_2V\left(\bold{s}_{1,4},\bold{s}'_{2,1}\right)
   +
   \\ (1-\mu_1)\mu_2(1-p_1)p_2V\left(\bold{s}_{1,4},\bold{s}'_{2,2}\right)
   + (1-\mu_1)(1-\mu_2)p_1p_2V\left(\bold{s}_{1,4},\bold{s}'_{2,3}\right)
   \\
   + (1-\mu_1)(1-\mu_2)(1-p_1)p_2 V\left(\bold{s}_{1,4},\bold{s}'_{2,4}\right)
   +(1-\mu_1)\mu_2p_1(1-p_2)V\left(\bold{s}_{1,4},\bold{s}_{2,1}\right)
   \\
   + (1-\mu_1)\mu_2(1-p_1)(1-p_2)V\left(\bold{s}_{1,4},\bold{s}_{2,2}\right)
   + (1-\mu_1)(1-\mu_2)p_1(1-p_2)V\left(\bold{s}_{1,4},\bold{s}_{2,3}\right)
   \\
   + (1-\mu_1)(1-\mu_2)(1-p_1)(1-p_2)V\left(\bold{s}_{1,4},\bold{s}_{2,4}\right),
    \end{array}
\end{equation}
where
$ 
\bold{s}'_{2,1}=(0,\tilde{\theta}_2,\tilde{x}_2),~\bold{s}'_{2,2}=(0,\tilde{x}_2+\tilde{\theta}_2,0),~\bold{s}'_{2,3}=(
\tilde{\theta}_2,0,\tilde{x}_2),~\bold{s}'_{2,4}=(
\tilde{\theta}_2,\tilde{x}_2,0).
 $ 
By the  above two equations, we have
\begin{equation}
\nonumber
    \begin{array}{ll}
    U\triangleq \Bbb{E}\left\{V(\bold{s}')~\big|~\bar{\bold{s}},(2,1)\right\}- \Bbb{E}\left\{V(\bold{s}')~\big|~\bar{\bold{s}},(2,2)\right\}= \mu_1\mu_2p_1 p_2\left[V\left(\bold{s}_{1,1},\bold{s}_{2,1}\right)-V\left(\bold{s}_{1,2},\bold{s}'_{2,1}\right) \right]
    \\
    + \mu_1\mu_2(1-p_1) p_2\left[V\left(\bold{s}_{1,1},\bold{s}_{2,2}\right)-V\left(\bold{s}_{1,2},\bold{s}'_{2,2}\right)\right]
     +
     \mu_1(1-\mu_2)p_1 p_2
        \\
     \left[V\left(\bold{s}_{1,1},\bold{s}_{2,3}\right)-V\left(\bold{s}_{1,2},\bold{s}'_{2,3}\right) \right]
     +
      \mu_1(1-\mu_2)(1-p_1) p_2
      \left[V\left(\bold{s}_{1,1},\bold{s}_{2,4}\right)-V\left(\bold{s}_{1,2},\bold{s}'_{2,4}\right) \right]
      +
            \\
      \left((1-\mu_1)\mu_2p_1 p_2\right)
      \left[V\left(\bold{s}_{1,3},\bold{s}_{2,1}\right)-V\left(\bold{s}_{1,4},\bold{s}'_{2,1}\right) \right]
      +
    (1-\mu_1)\mu_2(1-p_1) p_2
      \\
      \left[V\left(\bold{s}_{1,3},\bold{s}_{2,2}\right)-V\left(\bold{s}_{1,4},\bold{s}'_{2,2}\right) \right]
      +
      (1-\mu_1)(1-\mu_2)p_1 p_2
      \left[V\left(\bold{s}_{1,3},\bold{s}_{2,3}\right)-V\left(\bold{s}_{1,4},\bold{s}'_{2,3}\right) \right]
      +
      \\
          (1-\mu_1)(1-\mu_2)(1-p_1) p_2
      \left[V\left(\bold{s}_{1,3},\bold{s}_{2,4}\right)-V\left(\bold{s}_{1,4},\bold{s}'_{2,4}\right) \right].
    \end{array}
\end{equation}
Let  $G$ be the same as $U$, except for the following changes: ${\bold{s}_{1,2}\rightarrow (0,\tilde{x}_1+\tilde{\theta}_1,\tilde{y}_1)}$ and 
${\bold{s}_{1,4}\rightarrow(\tilde{\theta}_1,\tilde{x}_1,\tilde{y}_1)}$.
Now, we have 
 $ \Bbb{E}\left\{V(\bold{s}')~\big|~\bar{\bold{s}},(2,1)\right\}- \Bbb{E}\left\{V(\bold{s}')~\big|~\bar{\bold{s}},(2,2)\right\}= U \stackrel{(a)}{\le } G \stackrel{(b)}{\le } 0,$
where $(a)$ follows from the monotonicity  of $V(\bold{s})$ by Lemma \ref{Lemm_Non-dec} and $(b)$ follows from \eqref{Eq_Assum_NOndec}. 
We have shown that  a $\lambda$-optimal policy has the switching-type structure, 
which completes the proof. 
 \vspace{- 3 em}
\subsection{The Upper Bound for the Conditional Lyapunov Drift in \eqref{Eq_Drift_F}}\label{Appendix_UpperB}
To derive the upper bound for the conditional Lyapunov drift $\Delta[t]$, we use the following inequality in which, for any $A_1\ge0$, $A_2\ge0$, and $A_3\ge0$, we have \cite[p. 58]{Neely_Sch}
\begin{equation}\label{Eq_Ineq}
\begin{array}{ll}
 \left(\max\{A_1-A_2+A_3,0\}\right)^2\le A_1^2+A_2^2+A_3^2+2A_1(A_3-A_2).
 \end{array}
\end{equation}
By applying \eqref{Eq_Ineq} to the evolution of the virtual queue in \eqref{Eq_Evol_Q}, 
we obtain 
\begin{equation}\label{Eq_Drift}
    \begin{array}{ll}
         H^2[t+1]\le H^2[t]+\Gamma_{\max}^2+D(\bold{a}[t])^2+2H[t](D(\bold{a}[t])-\Gamma_{\max}).
    \end{array}
\end{equation}
By applying \eqref{Eq_Drift} to the 
conditional Lyapunov drift $\Delta[t]$ in \eqref{Eq_Drift_F},  we obtain
\begin{equation}\label{Eq_Drift_2}
    \begin{array}{ll}
         \Delta[t] & \le \Gamma_{\max}^2/2+ \Bbb{E}\{D(\bold{a}[t])^2~|~\mathcal{Z}[t]\}/2+H[t](\Bbb{E}\{D(\bold{a}[t])~|~\mathcal{Z}[t]\}-\Gamma_{\max})
        \\
        & \stackrel{(a)}{\le}  B + H[t](\Bbb{E}\{D(\bold{a}[t])~|~\mathcal{Z}[t]\}-\Gamma_{\max}),
    \end{array}
\end{equation}
where $B =  1/2\Gamma_{\max}^2+2$ and $(a)$ is due to using $\Bbb{E}\{D(\bold{a}[t])^2~|~\mathcal{Z}[t]\}\le 4$.
\vspace{- 1 em}
\subsection{Proof of Theorem \ref{Th_Stability}}\label{App_Stability}
To show the strong stability of the virtual queue under DPP-SP, first, we define an idle policy that chooses the idle decisions in  each slot $t$, i.e.,  $\alpha^{\mathrm{idl}}[t]=0$ and $\beta^{\mathrm{idl}}[t]=0$; hence,  $\bold{a}^{\mathrm{idl}}[t] \triangleq (0,0)$.
 By using inequality \eqref{Eq_upperB_DPP}, 
 we have
 \sloppy
\begin{align}\label{Eq_QS_1}
\sloppy
    \begin{array}{ll}
         \varphi[t] 
        & \stackrel{(a)}{\le}
          B+ H[t](\Bbb{E}\left\{D(\bold{a}[t])~\big|~\mathcal{Z}[t]\right\}-\Gamma_{\max}) 
          \\&
         + V \sum_i w_i \left(\Bbb{E}\left\{(1-\rho_2[t]\mathds{1}_{\{\beta[t]=i\}})y_i[t]+(1-\rho_1[t]\mathds{1}_{\{\alpha[t]=i\}})x_i[t] + x_i[t]+2\theta_i[t]+2~|~\mathcal{Z}[t]\right\}\right)
         \\ &
         \stackrel{(b)}{\le}
          B+ H[t](\Bbb{E}\left\{D(\bold{a}^{\mathrm{idl}}[t])~\big|~\mathcal{Z}[t]\right\}-\Gamma_{\max}) 
          \\&
         + V 
         \sum_i w_i \left(\Bbb{E}\left\{(1-\rho_2[t]\mathds{1}_{\{\beta^{\mathrm{idl}}[t]=i\}})y_i[t]+(1-\rho_1[t]\mathds{1}_{\{\alpha^{\mathrm{idl}}[t]=i\}})x_i[t] + x_i[t]+2\theta_i[t]+2~|~\mathcal{Z}[t]\right\}\right)
           \\ &
         \stackrel{(c)}{\le}
          B+ H[t] (\Bbb{E}\left\{D(\bold{a}^{\mathrm{idl}}[t])~\big|~\mathcal{Z}[t]\right\}-\Gamma_{\max})
         + V\sum_i w_i \left(y_i[t]+x_i[t] + x_i[t]+2\theta_i[t]+2\right),
    \end{array}
    \vspace{-1 em}
\end{align}
where 
$(a)$ is due to inequality \eqref{Eq_upperB_DPP}, 
$(b)$ follows because  DPP-SP, given by \eqref{Eq_MW_Policy_1}, minimizes  the upper bound of the drift-plus-penalty function, i.e., the R.H.S of $(a)$ in \eqref{Eq_QS_1}, in each slot $t$ among all the possible decisions, including  the idle decisions, 
and $(c)$ is due to the fact that, for any decisions in slot $t$, the inequalities ${\Bbb{E}\{1-\rho_1[t]\mathds{1}_{\{\alpha[t]=i\}}~|~\mathcal{Z}[t]\}\le 1}$ and ${\Bbb{E}\{1- \rho_2[t]\mathds{1}_{\{\beta[t]=i\}}~|~\mathcal{Z}[t]\}\le 1}$  hold. 
In \eqref{Eq_QS_1},
using the fact that ${\Bbb{E}\left\{D(\bold{a}^{\mathrm{idl}}[t])~\big|~\mathcal{Z}[t]\right\}=0}$ and the AoI values are bounded by finite $N$,
and taking expectations with respect to $\mathcal{Z}[t]$ and using the law of iterated expectations  yields:
\begin{equation}\label{Eq_QS}
    \begin{array}{ll}
        \Bbb{E}\left\{L(H[t+1])-L(H[t])+    V 
      {\textstyle\sum_i} \delta_i[t+1]
      \right\}
        {\le}
        B-
        \Gamma_{\max}
        \Bbb{E}\{H[t]\}
        + \underbrace{ V (5N+4) \textstyle\sum_i w_i }_{\triangleq \tilde{V}}
    \end{array}
\end{equation}
In \eqref{Eq_QS}, summing over $t=0,\dots, T-1$ (using the law of telescoping sums), dividing by positive $T$ and $\Gamma_{\max}$,
and rearranging  yields
\begin{equation}\label{Eq_Bound_VQ_1}
    \begin{array}{ll}
        \frac{1}{T}\sum_{t=0}^{T-1} \Bbb{E}\{H[t]\} &  \le \frac{B}{\Gamma_{\max}}
        -\frac{\Bbb{E}\{L(H[T])\}-\Bbb{E}\{L(H[0])\}}{T\Gamma_{\max}}
        - \frac{V}{\Gamma_{\max}}\frac{1}{T}\sum_{t=0}^{T-1}\textstyle{\sum_i} w_i \Bbb{E}\{ \delta_i[t+1]\}              
        +\frac{\tilde{V}}{\Gamma_{\max}}
        \\ & \stackrel{(a)}{\le}
        \frac{B+\tilde{V} }{\Gamma_{\max}} + \frac{\Bbb{E}\{L(H[0])\}}{T \Gamma_{\max}},
    \end{array}
\end{equation}
where $(a)$ follows because we neglected the negative terms in the L.H.S of $(a)$. By taking  a $\lim \sup$ of \eqref{Eq_Bound_VQ_1} as $T \rightarrow \infty$, and due to that $\Bbb{E}\{L(H[0])\}$ is finite, we obtain
\begin{align}
   \limsup_{T\rightarrow \infty}~
   \frac{1}{T}\textstyle\sum_{t=0}^{T-1} \Bbb{E}\{H[t]\} \le \frac{B~ +~ \tilde{V}}{\Gamma_{\max}},
\end{align}
which implies that the virtual queue is strongly stable.
\vspace{-1.5 em }

\bibliographystyle{ieeetr}

\bibliography{Bib/conf_short,Bib/IEEEabrv,Bib/Ref_2SRAoI}

\end{document}